\documentclass[a4paper,11pt]{article}
\usepackage[twoside,a4paper,margin=1in]{geometry}
\usepackage[utf8]{inputenc}
\usepackage{graphicx,url,amssymb,hyperref,amsmath,amsthm}
\usepackage{color}
\usepackage{xcolor}
\usepackage{tcolorbox}
\usepackage{thm-restate}
\usepackage{xspace}
\usepackage{amscd}
\usepackage{tikz}
\usepackage{framed}
\usepackage{float}
\theoremstyle{plain}
\newtheorem{theorem}{Theorem}[section]

\newtheorem{conjecture}[theorem]{Conjecture}
\newtheorem{corollary}[theorem]{Corollary}
\newtheorem{claim}[theorem]{Claim}
\newtheorem{lemma}[theorem]{Lemma}
\newtheorem{remark}[theorem]{Remark}
\newtheorem{definition}[theorem]{Definition}
\newtheorem*{remark2}{Important remark}

\title{Degenerate crossing number and signed reversal distance}

\title{Degenerate crossing number and signed reversal distance\thanks{This work was partially supported by the ANR project SoS (ANR-17-CE40-0033). An extended abstract appeared in the Proceedings of the 31st International Symposium on Graph Drawing and Network Visualization (GD 2023)}}

\author{Niloufar Fuladi\thanks{Université de Lorraine, CNRS, INRIA, LORIA, F-54000 Nancy, France} \and Alfredo Hubard\thanks{Univ Gustave Eiffel, CNRS, LIGM, F-77454 Marne-la-Vallée, France}\footnotemark[2] \and Arnaud de Mesmay\footnotemark[3]}
\begin{document}
\maketitle
\begin{abstract}
Given a graph drawn in the plane, the degenerate crossing number of the drawing is the number of points in the plane which are contained in the relative interior of at least two edges, where each edge is required to be drawn as a simple arc. The degenerate crossing number of a graph is the minimum degenerate crossing number among all its drawings.

Given a drawing, cutting a neighborhood of the surface around each crossing and pasting a M\"obius band gives a non-orientable surface, on which the drawing of the graph can be extended to an embedding. From this observation, Mohar derived that the degenerate crossing number of a graph is at most its non-orientable genus, and conjectured that these quantities are equal for every graph. He also made a stronger conjecture for loopless pseudo-triangulations with a fixed embedding scheme. 

In this paper, we prove a structure theorem that allows to understand when the degenerate crossing number and non-orientable genus coincide in a large class of loopless bipartite embedding schemes. In particular, we provide a counterexample to Mohar's stronger conjecture, but show that in the vast majority of the 2-vertex cases, as well as for many bipartite graphs, Mohar's conjecture is satisfied.

The reversal distance between two signed permutations is the minimum number of reversals that transform one permutation to the other one. 
If we represent the trajectory of each element of a signed permutation under successive reversals by a simple arc, we obtain a drawing of a 2-vertex embedding scheme with degenerate crossings. 
Our main result is proved by leveraging this connection and a classical result in genome rearrangement (the Hannenhalli--Pevzner algorithm) and can also be understood as an extension of this algorithm  when the reversals do not necessarily happen in a monotone order.

\end{abstract}

\section{Introduction}

A \emph{cross-cap drawing}\footnote{We are defining in this introduction all the terms required to understand our main results but refer to Section~\ref{preliminaries} for more background.} of a graph $G$ is a drawing of $G$ on the sphere with $g$ distinct points, called \emph{cross-caps}, such that the drawing is an embedding except at the cross-caps, where multiple edges are allowed to cross transversely, as pictured in Figure~\ref{mohar}. 
\begin{figure}
    \centering
    \includegraphics[width=.55\textwidth]{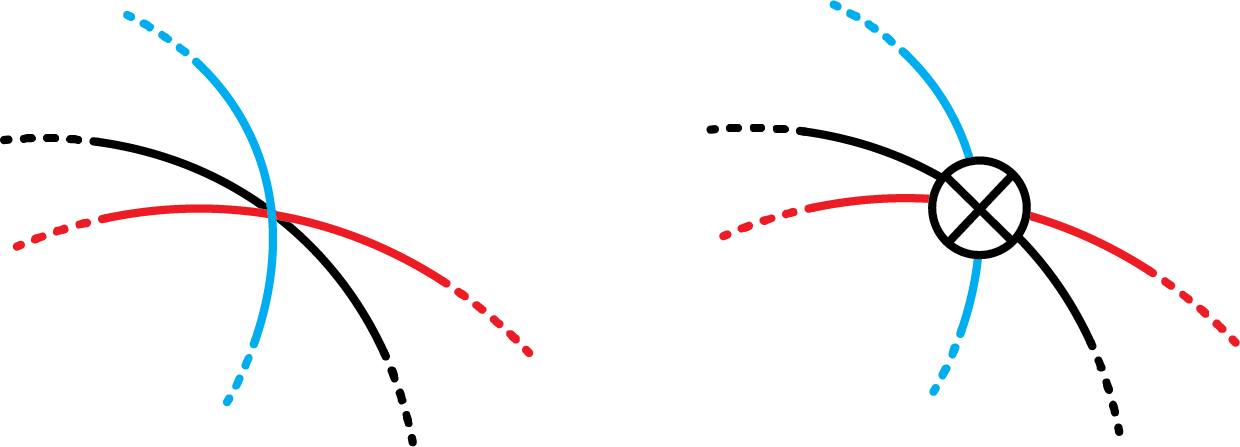}
\caption{A degenerate crossing and a cross-cap placed at this crossing.}
    \label{mohar}
\end{figure}
In~\cite{pach2009degenerate}, Pach and Toth introduced the \emph{degenerate crossing number}, denoted by $\mathrm{cr}_{\mathrm{deg}}$ which in this language is the minimum number of cross-caps for a cross-cap drawing of $G$ to exist, where edges are required to be drawn as simple arcs, that is, arcs without self-intersections. In~\cite{mohar2009genus}, Mohar removed the constraint that edges be simple arcs, leading to the \emph{genus crossing number}, which he proved to be equal to the \emph{non-orientable genus} of the graph, that is, the smallest possible genus of a non-orientable surface on which $G$ embeds. He then made an enticing conjecture claiming that these two crossing numbers are equal. Mohar's conjecture can be restated as follows: %

\begin{conjecture}~\cite[Conjecture~3.1, Proposition~3.3]{mohar2009genus}\label{second} Every simple graph $G$ of non-orientable genus $k$ has a cross-cap drawing with $k$ cross-caps where every edge enters each cross-cap at most once. 
\end{conjecture}

Mohar went further and conjectured that for any loopless graph embedding, there exists a cross-cap drawing with similar properties that is compatible with this embedding, in a sense that we now describe. The terminology that we introduce is equivalent to the PD1S in~\cite[Section~3]{mohar2009genus}. An embedding $\phi$ of a multi-graph on a non-orientable surface is \emph{cellular} if all the faces of the embedding are homeomorphic to disks. We denote by $g(\phi)$ the (non-orientable) genus of the underlying surface. A \emph{pseudo-triangulation} of a surface $S$ is a cellularly embedded multi-graph $\phi \colon G \to S$ in which each face has degree three. Denote by $S^2\setminus g\tiny{\bigotimes}$, the sphere minus $g$ tiny disks, and by $(S^2\setminus g\tiny{\bigotimes})/\sim$ the space obtained by quotienting the boundary of each disk with the antipodal map. Topologically, this amounts to gluing a M\"obius band on each missing disk, thus yielding the non-orientable surface of genus $g$, denoted by $N_g$. Given an embedded multi-graph $\phi \colon G \to N_g$, we say that $\phi'\colon G \to (S^2 \setminus g \tiny{\bigotimes})/\sim$ is a cross-cap drawing of $\phi$ if there is a homeomorphism $f \colon N_g \to (S^2 \setminus g \tiny{\bigotimes})/\sim$ such that $f(\phi(G))=\phi'(G)$. We say that a cross-cap drawing of a cellularly embedded multi-graph $\phi$ is \emph{perfect} if every edge enters each cross-cap at most once. Finally, we say that an embedded multi-graph $\phi$ \emph{admits a perfect cross-cap drawing} if there exists a cross-cap drawing $\phi'$ of $\phi$ that is perfect.

\begin{conjecture}~\cite[Conjecture~3.4]{mohar2009genus}\label{main} For any positive integer $g$, every loopless pseudo-triangulation of $N_g$ admits a perfect cross-cap drawing. 
\end{conjecture}

An even stronger conjecture was hinted at in~\cite[Paragraph following Conjecture~3.4]{mohar2009genus}, suggesting that one could possibly remove the loopless assumption if one forbids separating loops. This strengthening was disproved by Schaefer and \v{S}tefankovi\v{c}~\cite[Theorem~7]{JGAA-580}.

In addition to their motivation from crossing number theory, these conjectures would also shed light on the difficult task of visualizing high genus embedded graphs, providing an alternate approach to that of Duncan, Goodrich and Kobourov~\cite{duncan2011planar}, who rely on canonical polygonal schemes~\cite{lazarus2001computing}.

 A big step towards both these conjectures was achieved by Schaefer and \v{S}tefankovi\v{c}, who proved~\cite[Theorem~10]{JGAA-580} that 
any multi-graph embedded on a non-orientable surface of genus $g$ has a cross-cap drawing with $g$ cross-caps, in which each edge enters each cross-cap at most \emph{twice}. 
The theorem of Schaefer and \v{S}tefankovi\v{c} applies in particular to one-vertex embedded multi-graphs, 
and thus suggests a natural approach towards proving Conjectures~\ref{second} and~\ref{main}. 
First contract a spanning tree to obtain a one-vertex graph and apply this theorem. 
Then, edges might enter cross-caps twice, but since the initial graph is loopless, one could hope to uncontract some edges so as to spread these two cross-caps on two edges, thus obtaining a perfect cross-cap drawing. 
Our first result shows that this approach cannot work, as some loopless $2$-vertex graph embeddings do not admit perfect cross-cap drawings. 
\paragraph*{Our results}
A cellularly embedded multi-graph can be encoded by an \emph{embedding scheme}, also known as a rotation system, that consists of (1) the graph together with (2) a cyclic permutation of the half-edges incident to each vertex (3) a \emph{signature}, or \emph{sign}, for each edge. To ease the reading, in the entirety of this article we use the overline notation $\overline{i}$ to indicate negative signs, and no notation for the positive signs.

\begin{figure}[t]
\centering
    \includegraphics[width=\textwidth]{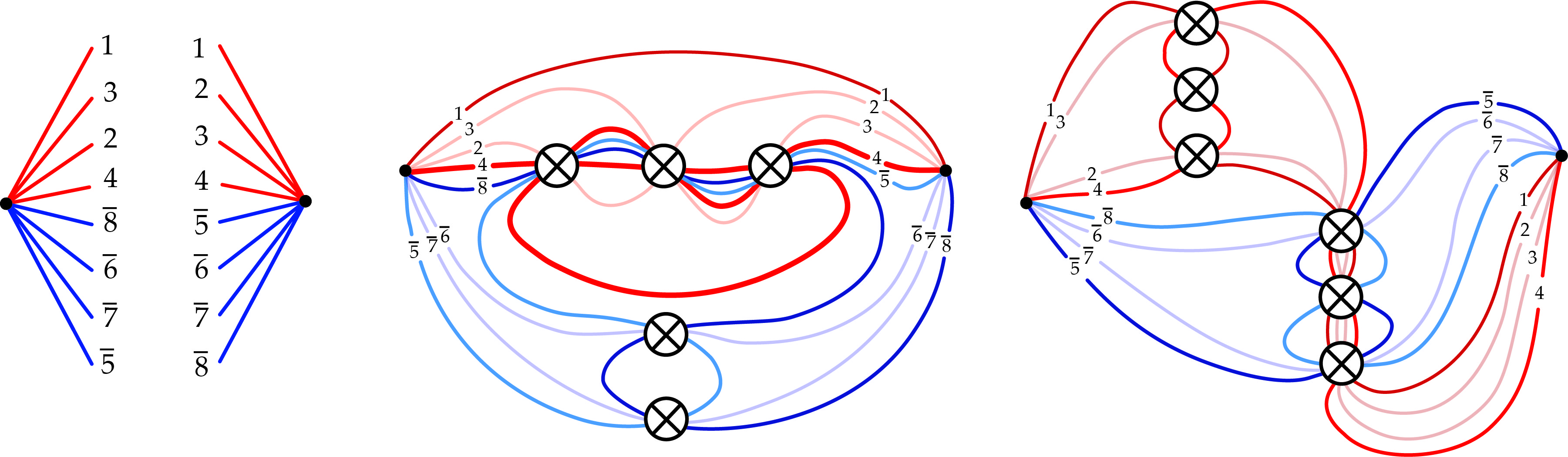}
    \caption{Left: A loopless 2-vertex scheme made of a \emph{positive block}, in red, consisting of only positive edges, and a \emph{negative block}, in blue, consisting of negative edges. Middle: a cross-cap drawing showing that it has non-orientable genus $5$. The bold red edge enters a cross-cap twice. Right: A cross-cap drawing where each edge enters each cross-cap at most once requires $6$ cross-caps. 
    }
    \label{block example}
\end{figure}

Given a 2-vertex loopless embedded graph, two edges whose union separates the surface into two components are \emph{separating}, and if one of the components is orientable while the other one is not orientable, we call the subgraph contained in the orientable component a \emph{block}. If this component is not a disk then we say it is \emph{non-trivial}. If the signs of all edges in this subgraph are \emph{positive} (resp. \emph{negative}) then we say that the block is \emph{positive} (resp. \emph{negative}). We refer to Figure~\ref{block example} for an example illustrating this notion of blocks. We refer to Sections~\ref{p:embedding scheme} and \ref{p:blocks} for more background on these definitions and a combinatorial reformulation.

\begin{restatable}{theorem}{permutation}\label{f}
 A loopless $2$-vertex embedding scheme that consists of exactly one non-trivial positive block and one non-trivial negative block admits no perfect cross-cap drawing. 
 
\end{restatable}

As a corollary, we obtain a counter-example to Conjecture~\ref{main}:

\begin{corollary}\label{maint}
There exists a loopless pseudo-triangulation $G$ that admits no perfect cross-cap drawing.
\end{corollary}

 Figure~\ref{doubled} shows that this counter-example is quite fragile by displaying an interesting phenomenon: while the example in Figure~\ref{block example} does not admit a perfect cross-cap drawing, surprisingly, it does after \emph{adding} two edges to it.

\begin{figure}[t]
    \centering
\includegraphics[width=.55\textwidth]{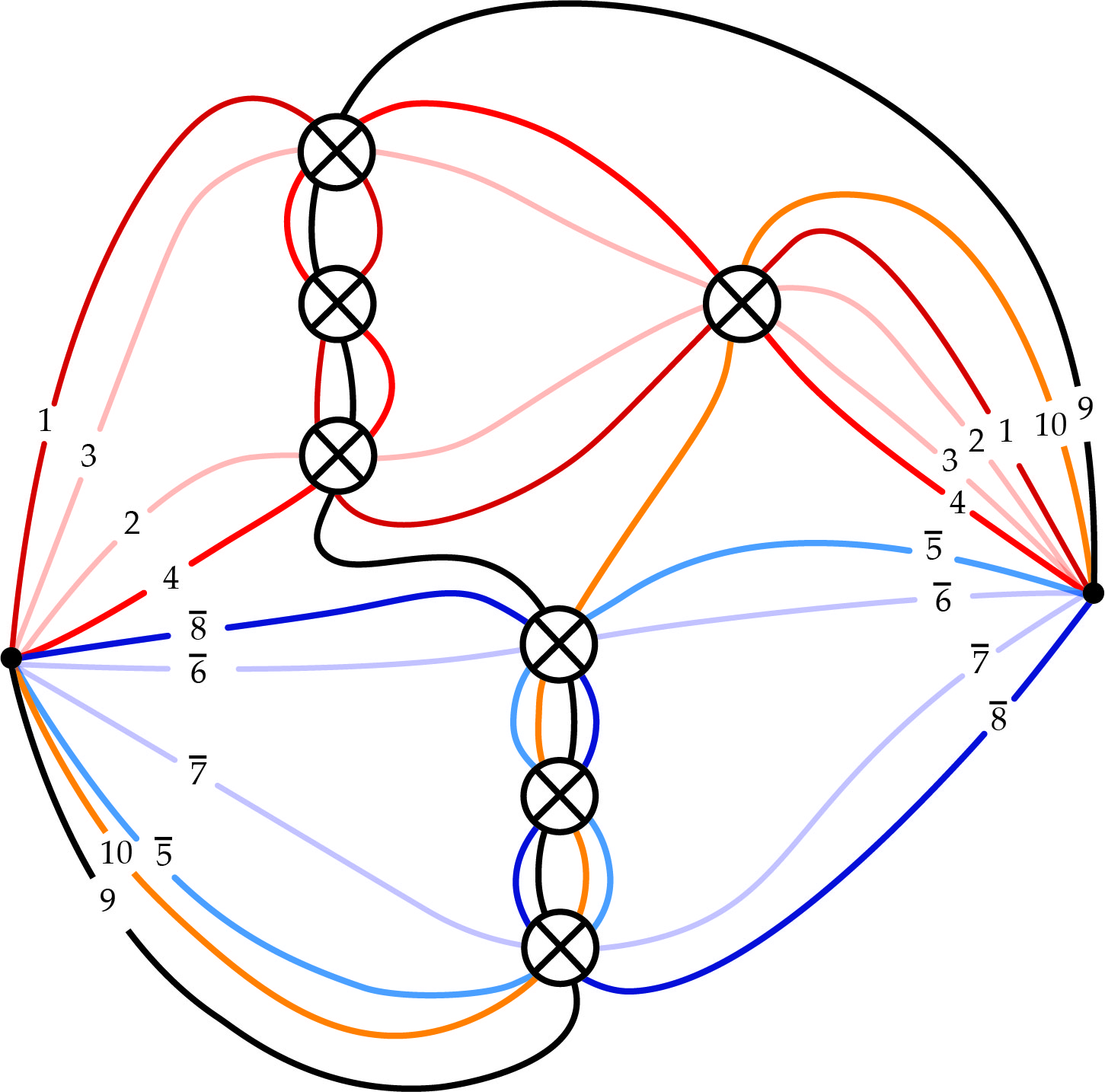}
    \caption{A perfect cross-cap drawing of Figure~\ref{block example} with two additional edges.}
    \label{doubled}
\end{figure}

\vspace{1em}

Our second contribution and main theorem is a converse to Theorem~\ref{f}.

In order to state it, we first introduce a bit of terminology. For $G$ a $2$-vertex loopless graph embedded on a surface $S$, we define its \emph{portions} as follows. We first cut along all the edges that belong to a pair of separating edges. This gives a decomposition of the surface into several connected components. If two connected components are non-orientable and have a common boundary, we paste them back together, and repeat until no two non-orientable sub-surfaces share a boundary. For each of the resulting subsurfaces, we then paste a disk on each boundary. We call the surfaces thus obtained, the \emph{portions} of the surface, and respectively, the intersections with the graph, the \emph{portions} of the graph. 
Given a portion $A \subset G$, the \emph{cleaned portion} $A|$ is an embedded graph obtained by iteratively simplifying \emph{homotopic edges}, i.e., repeatedly replacing pairs of edges $(a,b)$ occurring consecutively around both vertices by a single edge $a$ and likewise replacing pairs of edges $(\overline{a},\overline{b})$ occurring consecutively around both vertices by a single edge $\overline{a}$.  We refer to Figure~\ref{portions} for an illustration of these constructions. 

\begin{figure}[H]
    \centering
    \includegraphics[width=0.8\textwidth]{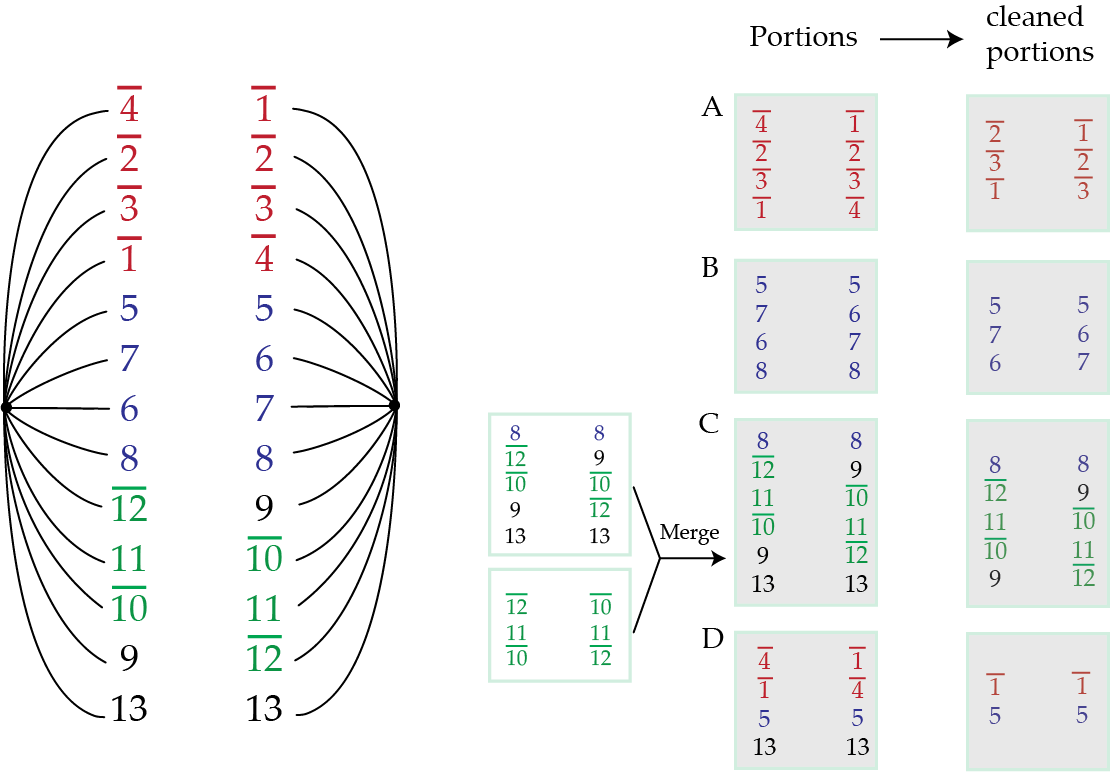}
    \caption{A decomposition of a loopless $2$-vertex embedded graph into its (cleaned) portions. This example is used as a running example and will reappear in Figures~\ref{portion tree} and~\ref{fig:corners}.}
    \label{portions}
\end{figure}

If $G$ is an embedding scheme of a bipartite graph, an \emph{allowable contraction} of $G$ is the embedding scheme that arises from the process of adding an edge in a face of $G$ between two vertices that are on the same side of the bipartition (and thus non-adjacent) and contracting it. A \emph{2-vertex contraction} of $G$ is an embedding scheme that arises from a maximal sequence of successive allowable contractions of $G$.

\begin{restatable}{theorem}{maintheorem}\label{main theorem}
For any embedding scheme $G$ of a bipartite graph on $N_g$, at least one of the following is true:  
 \begin{enumerate}
 \item $G$ admits a perfect cross-cap drawing,
 \item For every $2$-vertex allowable contraction of $G$, every non-orientable cleaned portion is of one of exactly two types: either (i) it is isomorphic to the left picture on Figure~\ref{reduced badly} or (ii) it is isomorphic to the right picture of Figure~\ref{reduced badly}.
  \end{enumerate}
\end{restatable}

This theorem shows that apart from some narrow families of exceptional cases, all bipartite embedded graphs satisfy Conjecture~\ref{main}. Observe that in the family of $2$-vertex counterexamples referred to in Theorem~\ref{f}, there are three portions: one orientable portion for each block and one non-orientable portion isomorphic to the left graph in Figure~\ref{reduced badly}.  

It is easy to see that Theorem~\ref{main theorem} implies that under standard random models, any loopless 2-vertex embedding scheme admits a perfect cross-cap drawing asymptotically almost surely. We believe that the same should hold for random bipartite embedded graphs. 

Theorem~\ref{main} is not exactly the strongest result that we can prove, a slightly stronger and more technical theorem is presented in Theorem~\ref{main theorem2} in Section~\ref{S:perfect}. We note that our proof techniques readily provide a polynomial-time algorithm that constructs the perfect cross-cap drawing when we are in the first case of the theorem.

\begin{figure}[H]
    \centering
    \includegraphics[width=.9\textwidth]{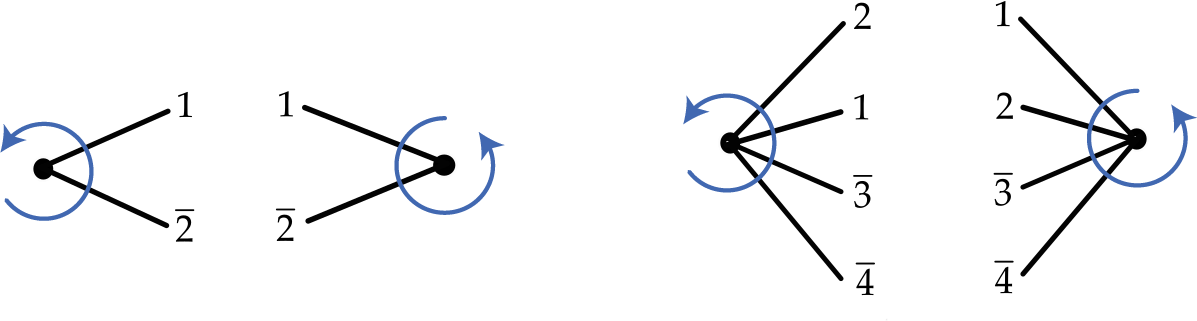}
    \caption{The only two obstructions to perfect cross-cap drawings. On the left, the two cyclic permutations are $(\bar{2},1)$ and $(1,\bar{2})$, and on the right, they are $(\bar{4},\bar{3},1,2)$ and $(1,2,\bar{3},\bar{4})$, where we use cycle notations and the $\bar{\cdot}$ denotes negative signature.}
    \label{reduced badly}
\end{figure}

 Most of the work in the proof of Theorem~\ref{main theorem} is about understanding the case of $2$-vertex graphs, and we focus on this case in most of the paper. We derive the general bipartite case from the $2$-vertex case in the last section.

\paragraph*{Techniques and connections to signed reversal distance}
Our focus on the 2-vertex loopless embedding schemes is further motivated by a connection (introduced in~\cite{fuladi2022short}) to computational genomics. 
An important problem in computational biology is to compute various notions of distance between two genomes (see~\cite{compeau2015bioinformatics,fertin2009combinatorics}). Remarkably, one of the most biologically relevant distances is also one of the few that can be calculated efficiently: a one chromosome genome is encoded by a \emph{signed permutation} $\pi=(\pi_1, \ldots, \pi_n)$. Here, signed means that each element has a $+$ or $-$ sign, which we also denote using overline notations.
The \emph{reversal} 
of the interval $(\pi_i,\pi_j)$
acts on $\pi$ by reversing the order of the elements  $\pi_i,\ldots, \pi_j$ as well as their signs, it maps
\[(\pi_1,\pi_2, \ldots, \pi_{i-1}, \pi_i,\pi_{i+1} \ldots,\pi_{j-1} ,\pi_j,\pi_{j+1}\ldots \pi_n),\] to 
\[(\pi_1,\pi_2, \ldots,\pi_{i-1}, \overline{\pi_j},\overline{\pi_{j-1}} \ldots,\overline{\pi_{i+1}},\overline{\pi_i}, \pi_{j+1}\ldots \pi_n).\]

 The \emph{reversal distance} between signed permutations $\pi,\pi'$, denoted by $d(\pi,\pi')$  is the minimum number of reversals needed to transform $\pi$ to $\pi'$. A celebrated algorithm of Hannenhalli and Pevzner~\cite{hannenhalli1999transforming} (see also~\cite{bergeron2001very,bergeron2004reversal}) computes in polynomial time the reversal distance between two signed permutations.

Switching back to topological graph theory, recall that an embedded graph can be encoded combinatorially with an embedding scheme, which amounts to specifying a cyclic permutation around each vertex and a sign for each edge. For 2-vertex loopless embedded graphs, we can number the edges so that the cyclic permutation around one vertex is $(1,2, \ldots , n)$. The starting observation of our work is that such an embedding is entirely defined by the cyclic permutation around the second vertex and the sign for each edge, i.e., the data of a \emph{signed cyclic permutation}. 

Now, given a signed permutation $\pi$, we can trace the trajectory of each element under the action of reversals in order to obtain a cross-cap drawing of that embedding scheme, where each reversal corresponds to a cross-cap and each edge is an x-monotone curve, and in particular no edge enters twice the same cross-cap (see Figure~\ref{trajectory} for an illustration). This easily implies the following inequality due to Bafna and Pevzner~\cite{bafna1996genome}: $g(\pi^c)\leq d(\pi,id)$, where $\pi^c$ is a \emph{cyclic} signed permutation induced\footnote{Precisely, the permutation $\pi$ needs to be mirrored in addition to be made cyclic, see Section~\ref{p:blocks}.} by $\pi$ and $g(\pi^c)$ is the non-orientable genus of the associated embedding scheme. It turns out that the inequality $g(\pi^c)\leq d(\pi,id)$  is central to the reversal distance theory, and the cases of equality are well understood. Our proof techniques heavily depend on the techniques that have been developed to understand these cases. In turn, our Theorem~\ref{main theorem} can be reinterpreted in the setting of signed permutations as providing an extension of the Hannenhalli--Pevzner theory applying beyond these equality cases.

\begin{figure}
    \centering
    \includegraphics[width=.8\textwidth]{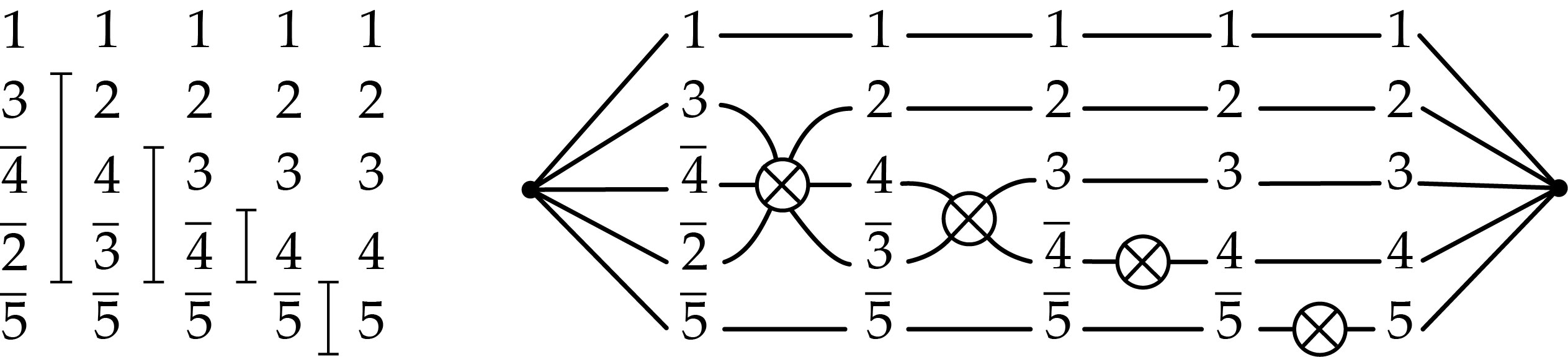}
    \caption{From sorting permutations with reversals to monotone cross-cap drawings.}
    \label{trajectory}
\end{figure}

The proof of Theorem~\ref{main theorem} in the $2$-vertex case consists of two steps which can readily be made algorithmic: when we are in the first case of the theorem, we show how to build a perfect cross-cap drawing of one of the non-orientable portions that is not of the two forbidden isomorphism classes. This is provided by Lemma~\ref{fantastique}. Then we devise a technique to inductively \emph{blow up} (Lemma~\ref{blowupcorrect}) the cross-cap drawing, adding portions one by one until we get a perfect cross-cap drawing of the original cyclic signed permutation.

\section{Preliminaries}\label{preliminaries}
\subsection{Embedding schemes}\label{p:embedding scheme}

We assume that the reader is familiar with basic concepts on the topology of surfaces (classification, genus, orientability, etc.) and refer to Mohar and Thomassen~\cite{mohar2001graphs} or Stillwell~\cite{stillwell1993classical} for background on these topics.

In this article, we work with multi-graphs, possibly with loops and multiple edges. To every edge of a multi-graph we associate two \emph{half-edges}, one incident to each endpoint of the edge: the main function of these half-edges is to differentiate the two incidences of a loop with a vertex. An \emph{embedding} of a graph $G$ on a surface $S$ is a continuous injective map $\phi:G \rightarrow S$. We consider two embeddings on a surface $S$ to be equivalent if there exists a homeomorphism of $S$ sending the image of one embedding to the other one. The \emph{faces} of an embedded graph are the connected components of $S \setminus \phi(G)$.

\begin{remark2}From now on, all the graphs that we consider are always embedded, and therefore we apply the common abuse of language to identify a graph $G$ with its embedding $\phi$.
\end{remark2}

 An embedding is \emph{cellular} if its faces are homeomorphic to topological disks. The Euler genus, $eg(G)$, of a cellular embedding of a graph $G$ is the quantity $2-v+e-f$, where $v$, $e$ and $f$ denote respectively the number of vertices, edges and faces of the embedding of $G$. If the surface $S$ that $G$ is embedded on is non-orientable and the embedding is cellular, the Euler genus $eg(G)$ equals the non-orientable genus of the surface that $G$ is embedded on, denoted by $g(G)$. 

A \emph{cyclic permutation} is a permutation consisting of a single cycle. If $S$ is orientable, a cellular embedding on $S$ can be described combinatorially by a \emph{rotation system}, that is, a set of cyclic permutations encoding the order of the half-edges around each vertex. When $S$ is non-orientable, which will almost always be the case in this paper, some additional data is required to encode a cellular embedding, which is encompassed in the concept of an \emph{embedding scheme}. We introduce the main definitions and refer to Mohar and Thomassen~\cite[Section~3.3]{mohar2001graphs} for extensive background.

\begin{definition}[Embedding scheme]\label{def:embedding scheme}
An \emph{embedding scheme} consists of a triple $(G,\rho,\lambda)$ where 
\begin{itemize}
    \item $G$ is a graph, 
    \item $\rho=\{\rho_v, v \in V(G)\}$, where each $\rho_v$ is a cyclic permutation of the half-edges incident to $v$, and
    \item $\lambda$ is a function that assigns a \emph{signature} (or just \emph{sign}) in $\{+1,-1\}$ to each edge of $G$.
\end{itemize}
\end{definition}

From an embedding scheme $(G,\rho,\lambda)$, one can naturally recover the set of \emph{facial cycles} by following the edges and switching sides according to their signatures. 
Then, pasting a topological disk on each facial cycle yields a cellular embedding of $G$ on a surface $S$. Given an embedding scheme $(G,\rho,\lambda)$, a \emph{flip} at a vertex $v$ is a transformation that yields another embedding scheme of the same graph, in which we reverse the order of the edges incident to $v$ and invert the signature of those edges incident to $v$ that are not loops. We say that two embedding schemes $(G,\rho,\lambda)$ and $(G,\rho',\lambda')$ are \emph{equivalent} if one can go from one to the other one by a sequence of flips. Two embedding schemes are equivalent if and only if they induce equivalent cellular embeddings~\cite[Theorem~3.3.1]{mohar2001graphs}. This justifies that equivalence classes of embedding schemes and embedded graphs can be considered as being two representations of the same objects, and we switch freely between the two points of view in this article, often using the shorthand $G$ to denote $(G,\rho,\lambda)$.

A \emph{cycle} in a graph is a closed walk without repeated vertices or edges, and a \emph{cycle} in an embedding scheme is a cycle on the edges of the underlying graph. A cycle in an embedding scheme is \emph{one-sided} (resp. \emph{two-sided}) if and only if the signatures of its edges multiply to $-1$ (resp. $+1$). An embedding scheme is called \emph{orientable} if it only contains two-sided cycles; otherwise it is called \emph{non-orientable}. A cycle $\gamma$ in a non-orientable embedding scheme is \emph{separating} if $S \setminus \gamma$ has two connected components. A cycle $\gamma$ in a non-orientable embedding scheme is \emph{orienting}~\cite{fuladi2022short,JGAA-580} if $S \setminus \gamma$ is a connected orientable surface. 

 We use cycle notation to describe cyclic permutations, and thus for example $(1,4,3,2)$ and $(4,3,2,1)$ denote the same permutation. The \emph{mirror} of a cyclic permutation is obtaining by reversing the order, e.g., the mirror of $(1,4,3,2)$ is $(2,3,4,1)$.
 An \emph{interval} $(a,b)$ in a cyclic permutation $\rho=(\rho_ 1, \ldots , \rho_n)$ with $n$ elements is the sequence of elements occurring between the element $a$ and the element $b$ in $\rho$, i.e., the subpermutation $(\rho_i, \rho_{i+1}, \ldots, \rho_{j-1},\rho_j)$ for $i$ and $j$ such that $\rho_i=a$ and $\rho_j=b$ (where the indices are in $\mathbb{Z}_n$). Note that the intervals $(a,b)$ and $(b,a)$ are different. See Figure~\ref{intervals} for an example. For a cyclic permutation $\rho=(\rho_ 1, \ldots , \rho_n)$, we denote by $\overline{\rho}$ the mirrored cyclic permutation $\rho=(\rho_ n, \ldots , \rho_1)$.

\begin{figure}
    \centering
    \includegraphics[width=.28\textwidth]{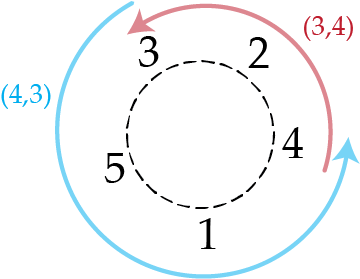}
    \caption{In the cyclic permutation $\pi=(3,2,4,1,5)$, the interval $(3,4)$ consists of the elements $\{3,2,4\}$  and the interval $(4,3)$ consists of the elements $\{4,1,5,3\}$.
}
    \label{intervals}
\end{figure}

 We use a different interval notation $[a,b]$ to denote all the integers between $a$ and $b$ modulo $n$, i.e., this is the standard interval $[a,b]$ if $a<b$ but it is the set $\{b,b+1, \ldots n, 1, \ldots a\}$ if $a>b$. Equivalently, $[a,b]$ is the interval $(a,b)$ in the cyclic permutation $(1, \ldots, n)$.

\subsection{2-vertex loopless embedding schemes and signed cyclic permutations}\label{p:blocks}

This paper almost exclusively deals with loopless graphs with two vertices. In that setting, without loss of generality, one can number the edges so that the cyclic permutation around one of the vertices is the cyclic permutation $(1,2, \ldots ,n)$, which we denote by $I$. Then the data of the embedding scheme just consists of the cyclic permutation around the other vertex, and the signature of the edges, and thus this amounts to a \emph{signed cyclic permutation}: a cyclic permutation where each number is additionally endowed with a $+$ or $-$ sign, which we denote using the overline notation $\overline{i}$.

\begin{remark2} Throughout the rest of the article, we freely identify a signed cyclic permutation $\pi$ and the 2-vertex loopless embedding scheme where:
\begin{itemize}
    \item the cyclic permutation around one vertex is $I$, 
    \item the cyclic permutation around the other vertex is the \textbf{mirror} of $\pi$, and 
    \item the signatures of the edges are indicated by the signs of the elements in $\pi$.
\end{itemize}
There is a mirror here in order to be coherent with signed reversal distance: this ensures that the base case $\pi=I$ corresponds to an orientable, planar (genus zero) embedding scheme. So with this identification, the two minimal obstructions pictured in Figure~\ref{reduced badly} are the signed cyclic permutations $(1,\bar{2})$ and $(2,1,\bar{3},\bar{4})$, i.e., we can read them from the figures from top to bottom.

Throughout the paper, all the permutations we work with are cyclic, and therefore all the index notations are to be understood modulo the number of elements in the permutation. We use \emph{sc-permutation} as a short-hand for signed cyclic permutation.
\end{remark2}

We denote by $eg(\pi)$ the Euler genus of the embedding scheme associated to an sc-permutation $\pi$.

\paragraph*{Topology of cycles in 2-vertex loopless embedding schemes} Denoting the vertices of a $2$-vertex loopless embedding scheme by $v_1$ and $v_2$, we define the \emph{wedges} $\omega_{a,b}$ and $\omega_{b,a}$ between $a$ and $b$ as follows:
\begin{itemize}
\item If at least one of $a$ and $b$ is positive, then $\omega_{a,b}$ (resp. $\omega_{b,a}$) contains all the half-edges in the interval $(b,a)$ (resp. $(a,b)$) in $\rho_{v_1}$ and $(a,b)$ (resp. $(b,a)$) in $\rho_{v_2}$.

    \item If both $a$ and $b$ are negative, then $\omega_{a,b}$ (resp. $\omega_{b,a}$) contains all the half-edges in the interval $(a,b)$ (resp. $(b,a)$) in both $\rho_{v_1}$ and $\rho_{v_2}$. 

\end{itemize} 

We say that a wedge \emph{encloses} an edge if it contains both its half-edges or none of them. For example $w_{1,4}$ in Figure~\ref{wedge} encloses all the edges of the graph. Observe that, following the definitions, $\omega_{a,b}$ encloses an edge $c$ if and only if $\omega_{b,a}$ does. 

For two edges $a$ and $b$ in a 2-vertex loopless embedding scheme, we denote by $a \cdot b$ the concatenation of $a$ and $b$ which we will interpret as a cycle. We can identify orienting and separating cycles in an embedding scheme based on whether or not they enclose edges in the embedding scheme. Since $\omega_{a,b}$ and $\omega_{b,a}$ are equivalent in terms of whether they enclose edges, we state the lemma only for one of them.

\begin{figure}
    \centering
    \includegraphics[width=\textwidth]{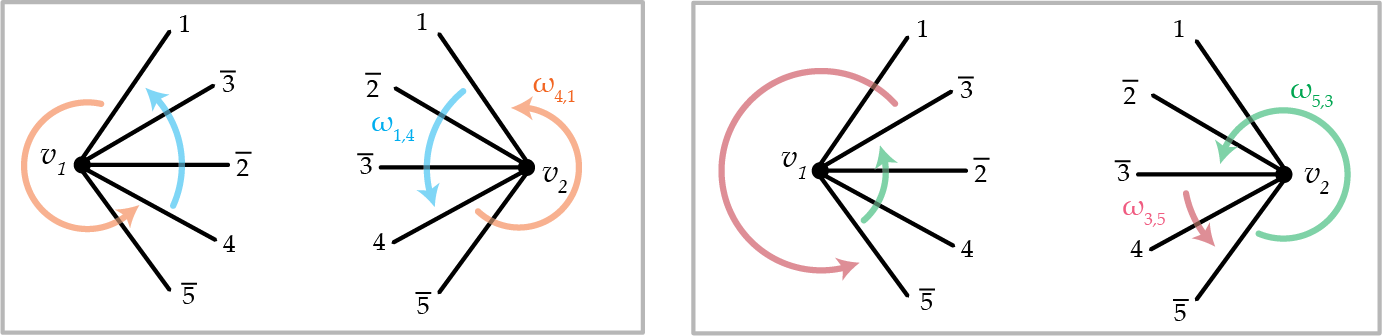}
    \caption{The wedges $\omega_{1,4}$, $\omega_{4,1}$, $\omega_{3,5}$ and $\omega_{5,3}$. The curve $1\cdot 4$ is separating and the curve $3\cdot 5$ is orienting. 
}
    \label{wedge}
\end{figure}

\begin{restatable}
    {lemma}{wedgelemma}
\label{curves}
 For two edges $a$ and $b$ in a non-orientable loopless 2-vertex embedding scheme:

\begin{itemize}
    \item If at least one of $a$ and $b$ is positive, the cycle $a \cdot b$ is orienting if and only if $\omega_{a,b}$ contains exactly one half-edge of each negative edge and encloses all the positive edges.
\item If both $a$ and $b$ are negative, the cycle $a \cdot b$ is orienting if and only if $\omega_{a,b}$ contains exactly one half-edge of each positive edge and encloses all the negative edges.
\end{itemize} 
\end{restatable}

\begin{proof}
The lemma is proved by contracting an edge to obtain a $1$-vertex scheme. We first introduce the following notation: for two cyclic permutations each containing a unique common letter $e$, written in cycle notation $\rho^1=(\rho^1_1, \rho^1_2, \ldots, \rho^1_{n_1-1}, e)$ and $\rho^2=(e,\rho^2_2, \ldots ,\rho^2_{n_2})$, the concatenation $\rho=\rho^{1}\overset{\tiny{e}}{\cdot}\rho^{2}$ is the cyclic permutation obtained by concatenating them at $e$: $(\rho^1_1, \ldots \rho^1_{n_1},\rho^2_2, \ldots , \rho^2_{n_2})$. When we contract a positive edge $e=v_1v_2$ in $G$, we obtain an embedding scheme $G'$ with a single vertex $w$ such that the cyclic permutation of the edges around $w$ after contraction is $\rho_w=\rho_{v_1}\overset{\tiny{e}}{\cdot}\rho_{v_2}$ where the two half-edges of $e$ have been identified as a unique common letter. On the other hand, when we contract a negative edge $e=v_1v_2$ in $G$, $\rho_w=\rho_{v_1}\overset{\tiny{e}}{\cdot}\overline{\rho_{v_2}}$, where the two half-edges of $e$ have been identified as a unique common letter and the signature of all the edges are reversed.

We prove the lemma by contracting the edge $a$ in order to obtain a one-vertex embedding scheme $G'$.
Note that the topological type of the cycle formed by the edges $a$ and $b$ in $G$ is the same as that of the loop $b$ in $G'$. 
By~\cite[Lemma~2.3]{fuladi2022short}, a loop $o$ in a 1-vertex non-orientable embedding scheme is orienting if and only if its halves \emph{alternate} with the halves of all negative loops in the cyclic permutation around the vertex and do not alternate with the halves of any positive loop.

To prove the first case we can assume that $a$ is positive, exchanging the roles of $a$ and $b$ if necessary. We contract the edge $a$ in $G$. The two halves of the loop $b$ subdivide the half-edges around $w$ into two sets, and the definition of wedge implies that the half-edges in $\omega_{a,b}$ in $G$ correspond to one of the sets of half-edges divided by $b$ in $\rho_w$. Since the wedge $\omega_{a,b}$ in $G$ contains exactly one half-edge of each negative edge, the halves of the loop $b$ alternate with the halves of negative loops in $\rho_w$. Similarly, since $\omega_{a,b}$ encloses all the positive edges, the halves of $b$ do not alternate with the halves of any positive loop in $\rho_w$. Therefore $b$ is orienting in $G'$ and this implies that the cycle formed by $a$ and $b$ is orienting in $G$. 

For the proof of the second case we proceed similarly by contracting the negative edge $a$. Compared to the previous case, this yields three changes, any two of which cancel each other: the definition of wedge changes, the signatures of all the edges are reversed during the concatenation, and the permutation $\rho_{v_2}$ is mirrored. This results in the condition in that case being the opposite of the one in the first case.
\end{proof}

Similarly, the following lemma allows us to recognize separating cycles combinatorially.

\begin{restatable}{lemma}{wedgelemmaseparating}\label{curves2}
    For two edges $a$ and $b$ in a loopless 2-vertex embedding scheme, the cycle $a \cdot b$ is separating if and only if $a$ and $b$ have the same signature and $\omega_{a,b}$ encloses all the edges.
    
\end{restatable}

\begin{proof}
We contract $a$. If $a$ is positive, then $b$ has the same signature and this is maintained after contraction. If $a$ is negative, then $b$ is also negative but this is reversed by the contraction. In both cases, after the contraction, the loop $b$ has positive signature in $G'$. Furthermore, the half-edges in $\omega_{a,b}$ in $G$ correspond to one of the sets of half-edges divided by $b$ in $G'$ (in the case where $a$ is negative, the change in the definition of wedge and the mirroring in the contraction cancel out). Therefore, for every loop in $G'$, the two halves are contained on the same side of $b$. Since $b$ is positive, it is two-sided, and thus it partitions the set of faces into two-nonempty subsets. Therefore $b$ is separating. Conversely if the halves of some edge alternate with $b$, this contradicts the fact that $b$ is separating since one can use that edge to connect its two sides. Uncontracting $a$ yields the proof of the lemma.
\end{proof}

See Figure~\ref{wedge} for an example of a separating and an orienting cycle in a 2-vertex scheme.

\vspace{1em}

\paragraph*{Blocks, homotopies and portions} 
For an sc-permutation $\pi$ which we interpret as a 2-vertex embedding scheme, applying a flip at one of the vertices amounts to \emph{flipping} $\pi$, i.e., mirroring the sc-permutation and changing the sign of all its elements. We say that two sc-permutations are equivalent if they differ by such a flip.

A \emph{positive block} in an sc-permutation $\pi$ is an interval $I=(\pi_i, \pi_j)$ where all the elements are positive, and $I$ contains exactly all the elements in $[\pi_i,\pi_j]$. For example the interval $(6,3)=\{6,2,1,3\}$ is a positive block in the sc-permutation $(2,1,3,\bar{5},4,6)$. A \emph{negative block} in an sc-permutation is an interval $I=(\pi_i, \pi_j)$ where all the elements are negative, and $I$ contains exactly all the elements in $[\pi_j,\pi_i]$, for example $(\bar{6},\bar{3})$ is a negative block in the sc-permutation $(2,1,\bar{6},\bar{4},\bar{5},\bar{3},7)$. A positive (respectively a negative) block is \emph{trivial} if it is sorted in increasing (respectively decreasing) order, i.e., if it is equal to $(\pi_i, \pi_{i}+1, \ldots ,\pi_{j}-1,\pi_j)$ or to $(\overline{\pi_i}, \overline{\pi_i-1}, \ldots , \overline{\pi_{j}+1},\overline{\pi_j})$. 
In these examples and also in non trivial blocks we call $\pi_i$ and $\pi_j$ the \emph{frames} of the block.
Two edges are \emph{homotopic} if they belong to a common trivial block (note that this matches the definition given in the introduction). A block is (inclusionwise) \emph{minimal} if it does not contain any block except itself. We remark that the notion of blocks is similar to the notions of hurdles of~\cite{hannenhalli1999transforming}, and of unoriented components in~\cite{bergeron2004reversal}. 

From the point of view of embedding schemes, a block corresponds to a collection of edges defining an orientable subscheme, as defined in the introduction. Indeed, by Lemma~\ref{curves2} the condition of containing exactly all integers between the positive integer values of the two frames $\pi_i$ and $\pi_j$, together with the frames having the same sign, corresponds to the closed curve $\pi_i\cdot \pi_j$ being a separating curve. The condition on all the edges in the block having the same sign implies that all closed curves in the block are two-sided and therefore the frames separate an orientable portion.

Recall from the introduction that \emph{portions} are defined by (i) cutting a $2$-vertex embedding scheme $G$ along all edges belonging to separating curves, yielding connected components of the surface, (ii) merging together adjacent non-orientable components and (iii) capping off boundaries with disks. These portions are naturally arranged in a tree.  

To describe the portions tree, we first consider the connected components before the merging at step (ii) takes place.  Define a relation between edges, $e\sim e'$ if $e \cdot e'$ is a separating curve. It follows from Lemma~\ref{curves2} that this is an equivalence relation. Moreover, by the same lemma, for each of these equivalence classes, the restrictions of the sc-permutation $\pi$ and the sc-permutation $I$ to elements of that class are identical. 

We now define a bipartite graph $T'$, that has one node for each connected component appearing after step (i) above, and one node for each equivalence class described in the previous paragraph. A node in $T'$ corresponding to a connected component $C$ is connected to a node in $T'$ corresponding to an equivalence class $Q$ if a pair of edges on the boundary of $C$ forming a separating curve belongs to $Q$. It is immediate that $T'$ is a tree since every edge of $T'$ corresponds to a separating curve of $G$, and removing it separates $T'$.

Finally we define the \emph{portions tree} $T$ to be a quotient of $T'$:  we merge any two nodes of $T'$ corresponding to non-orientable connected components that share an edge. This implies that they are at distance two in $T'$, and thus this merging preserves the tree structure. Furthermore, if $v$ is a node corresponding to an equivalence class that has degree $2$ in $T'$ after this merging, connecting a node $u$ to a node $w$, we erase $v$ and connect directly $u$ to $w$. The result of these operations is the portions tree $T$. Each of the portions that we had previously defined is a node of $T$, and there are some additional nodes corresponding to equivalence classes of edges connecting these portions together. We say that an edge of $G$ is a \emph{boundary edge} if it is involved in an edge of the portions tree, i.e., either if it is contained in a separating curve between two different vertices of the portions tree, or if its equivalence class is present in the portion tree. We refer to Figure~\ref{portion tree} for an illustration.

\begin{remark}\label{spqr}
This structure is similar to the construction of SPQR-trees~\cite{di1989incremental,di1990line}, which are the trees describing the structure of $3$-connected components in a $2$-connected graph. We chose the above description to be self-contained, but the connection is as follows. Let us define a graph $H$ to be a graph that has as vertex set the edge set of $G$, and as edge set the pairs of edges incident to a common face, i.e., $V(H)=E(G)$, and $(e,e') \in E(H)$ if $e$ and $e'$ are incident to a common face in $G$. Observe that since no edge in $G$ can separate a surface by itself, the graph $H$ is $2$-connected.

The connected components obtained after cutting along all edges corresponding to separating curves in $G$ correspond to the $3$-connected components of $H$ (the $R$ nodes of the SPQR tree), which are similarly obtained after cutting along all vertices belonging to $2$-vertex cuts. The nodes corresponding to equivalence classes of degree at least three are the $S$ nodes in the SPQR tree. Since a separating curve is adjacent to exactly two components, the $P$ nodes in the SPQR tree have degree exactly two and thus it is natural to dissolve them and connect directly their endpoints. Then, the portions tree is obtained by merging non-orientable $R$ nodes which are either adjacent in the SPQR tree, or adjacent to two consecutive edges of an $S$ nodes (this is the case where they share exactly one edge). See Figure~\ref{portion tree} for an illustration. 
\end{remark}

A portion is \emph{orientable} if the subsurface that it corresponds to is orientable, and \emph{non-orientable} otherwise. Note that minimal blocks are orientable portions which are leaves in the portions tree. 

\begin{figure}[H]
    \centering
    \includegraphics[width=\textwidth]{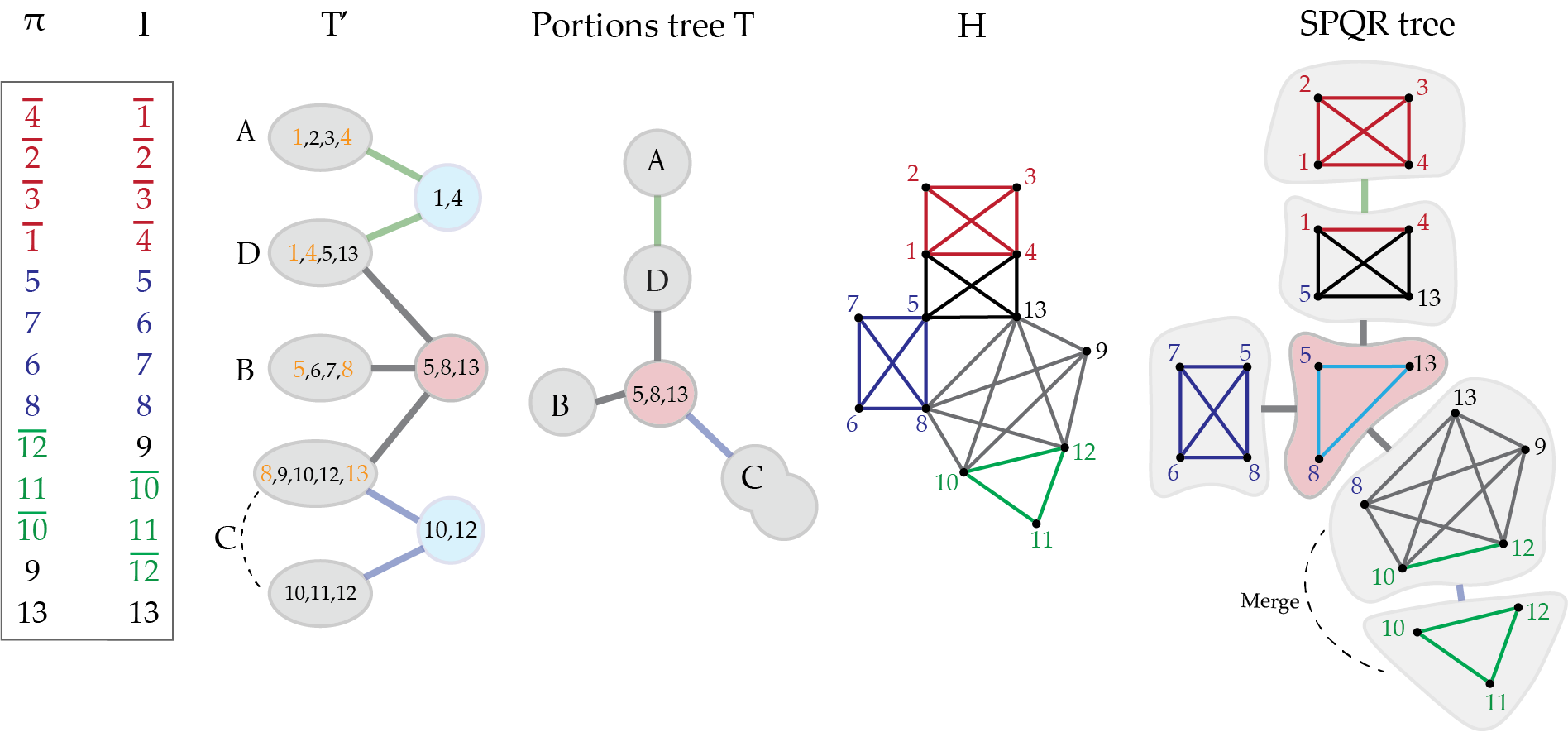}
    \caption{The embedding scheme corresponding to $\pi$, the bipartite graph $T'$, in which left nodes correspond to connected components and the right nodes correspond to equivalence classes of edges and the portions tree $T$ in which two of the right vertices in $T'$ have been erased. Note that the boundary edges are $1,4,5,8,13$ colored in yellow. The two figures at the right depict the graph $H$ and the corresponding SPQR tree that gives the same construction for the portions tree $T$ by Remark~\ref{spqr}. Note that the gray nodes correspond to $R$ nodes and the pink node corresponds to an $S$ node.}
    \label{portion tree}
\end{figure}

Finally, recall from the introduction that each portion can be \emph{cleaned}, by identifying all homotopic edges. We say that a $2$-vertex embedding scheme is \emph{reduced} if its portions tree consists of a single vertex, and it is immediate to see that the embedding scheme corresponding to a cleaned portion is reduced.

\subsection{Cross-cap drawings}

 Recall from the introduction that a cross-cap drawing of an embedded graph $G$ is a map $\phi': G \rightarrow (S^2\setminus g\tiny{\bigotimes})/\sim$ such that there exists a homeomorphism $f\colon N_g \to (S^2\setminus g\tiny{\bigotimes})/\sim$, such that $f(\phi(G))=\phi'(G)$, where $\phi$ is the embedding of $G$. Note that this is also well-defined for non-cellular embeddings. If the embedding $G$ is cellular, by Euler's formula, the minimum number of cross-caps for a cross-cap drawing of $G$ coincides with the Euler genus.

 The next lemma allows us to recognize types of cycles in a cross-cap drawing.

\begin{lemma}\label{separating}
   For any cross-cap drawing of a non-orientable embedding scheme:
   
   \begin{enumerate}
        \item A cycle is one-sided (two-sided) if and only if it enters an odd (even) number of cross-caps. 
   
       \item A cycle is orienting (separating)  if and only if it enters each cross-cap an odd (even) number of times. 
    
   \end{enumerate}
   
\end{lemma}

    We refer to~\cite[Lemma~3]{JGAA-580} and~\cite[Lemma~4]{JGAA-580} for proofs.

The definition of cross-cap drawing above requires the underlying surface to be orientable and thus does not apply to orientable embedding schemes. We extend it to orientable embedding schemes of positive Euler genus as follows. A \emph{non-orientable extension} of an orientable embedding scheme $G$ is a (non-cellular!) embedding obtained by adding exactly one cross-cap in one of the faces of $G$. This choice is in general not unique. A cross-cap drawing of an orientable embedding $G$ is a cross-cap drawing of one of its non-orientable extensions. Such a cross-cap drawing has $eg(G)+1$ cross-caps. 

Accordingly, when $G$ is an orientable embedding scheme with non-zero Euler genus, we define its non-orientable genus to be $g(G):=eg(G)+1$.  As for the Euler genus, this definition and notation is also used for sc-permutations. As in the non-orientable case, a cross-cap drawing is \emph{perfect} if every edge enters every cross-cap at most once. 

\subsection{Sc-reversal distance and monotone cross-cap drawings}\label{cross-caps}

Signed permutations model genomes with a single chromosome in computational biology where they come endowed with the reversal distance. Since our focus is on sc-permutations, we directly work throughout this article with a cyclic analogue of the classical theory. 
For an sc-permutation $\pi$, the \emph{signed cyclic reversal} (sc-reversal) of the interval $(\pi_i,\pi_j)$ acts on $\pi$ by reversing the order of elements in the interval as well as their signs. Recall that $I$ denotes the sc-permutation $(1, \ldots, n)$. The sc-reversal distance $d^c(\pi,I)$ is the smallest number $d$ such that there exists a sequence $\{\pi^1,\pi^2, \ldots, \pi^{d+1}\}$ such that $\pi^1=\pi$ or its flip, $\pi^{d+1}=I$ or its flip and for each $i$ in $[1,d]$, $\pi^i$ and $\pi^{i+1}$ differ by an sc-reversal. We call such a (not necessarily minimizing) sequence, a \emph{path} of sc-permutations. Note that if $\pi^{d+1}$ is the flip of $I$, by applying a flip on each sc-permutation in the path, we can use the same set of sc-reversals to go from the flip of $\pi^1$ to $I$ so we focus in the case $\pi^{d+1}=I$ in this paragraph.

To any path of sc-permutations $\{\pi^1,\pi^2,\ldots, \pi^{d+1}\}$ we can associate a cross-cap drawing. The construction is very similar to \emph{wiring diagrams} of \emph{allowable sequences of permutations}, the only difference is that we work on a cylinder except of working on a strip. We now explain this construction, see the picture in Figure~\ref{trajectory}.

In a path of sc-permutations,
$\pi^1,\pi^2, \ldots, \pi^{d+1}=I$, consecutive permutations are related by the action of an sc-reversal, which we denote by $r_i$, i.e., $\pi^i \cdot r_i =\pi^{i+1}$.

Each sc-reversal is a (non-cyclic) permutation, and we write 
\[R_k=r_{d-k}r_{d-k+1}\ldots r_{d},\]
and think of $R_k(j)$ as the ``position" of the edge $j$ in $\pi^k$. This makes sense if we think of $\pi^{d+1}=I$ as being written $(1, \ldots, n)$ in cycle notation, and then we have $R_0(j)=j$, and thus the position of the edge $j$ in $\pi^{d+1}=I$ is simply $j$, as it is in $(1, \ldots, n)$. To make this precise we work in the cylinder $[-1,d+2] \times [0,-n] / \sim$, where for all $x$ in $[-1,d+2]$, $(x,0)\sim (x,-n)$, and the indices in the second coordinate are understood modulo $n$. Then, we quotient each boundary to a point, yielding two distinct points: an initial vertex at $(-1,-n/2)$ and a terminal vertex at $(d+2,-n/2)$. Now  edges will be $x$-monotone piece-wise linear curves between these two vertices. The edge $j$ first emanates from $(-1,-n/2)$ and connects to $(0,R_1(j))$ with a straight-line. Then it goes to $(i,R_i(j))$ for each $i$ in $[1,d+1]$ and finally it connects $(d+1,j)$ to the terminal vertex at $(d+2,-n/2)$ with a straight-line between consecutive integer $x$-coordinates. The path taken by an edge $j$ between $(i,R_i(j))$ and $(i+1,-R_{i+1}(j))$ depends on the sc-reversal between these two sc-permutations. There are two possible choices that wind strictly less than once around the cylinder: either a straight line in the rectangle $[-1,d+2] \times [0,-n]$ or a path that wraps around the cylinder. We use the following rule: if the sc-reversal contains $n$ and $1$, then the edge is drawn as wrapping around the cylinder, otherwise it is drawn as a straight line. For each sc-reversal we can draw all the edges involved in the reversal with a single (degenerate) crossing, and finally we put a cross-cap at the crossing. Figure~\ref{fig:cylinder} illustrates the minimizing path between $\pi=(\bar{6},2,\bar{5}, 4,\bar{3},\bar{1})$ and the identity.  \begin{align*}
    \pi^1&=(\pi_1^1,\pi^1_2,\pi_3^1, \pi_4^1, \pi_5^1,\pi^1_6)=(\bar{6},2,\bar{5}, 4,\bar{3},\bar{1})\\ 
    \pi^2&=(\pi_1^2,\pi^2_2,\pi_3^2, \pi_4^2, \pi_5^2,\pi^2_6)=(\bar{6},2,3,\bar{4},5,\bar{1})\\
    \pi^3&=(\pi_1^3,\pi^3_2,\pi_3^3, \pi_4^3, \pi_5^3,\pi^3_6)=(1,2,3,\bar{4},5,6)\\
    \pi^4&=(\pi_1^4,\pi^4_2,\pi_3^4, \pi_4^4, \pi_5^4,\pi^4_6)=(1,2,3,4,5,6)=I.
    \end{align*}

\begin{figure}
    \centering
    \includegraphics[width=.45\textwidth]{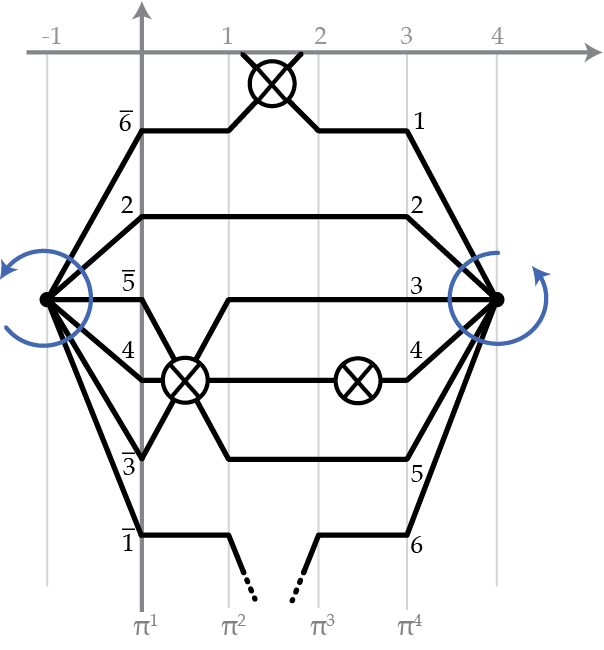}
    \caption{A minimizing path of sc-permutations.}
    \label{fig:cylinder}
\end{figure}

 Our convention is that in all our drawings, the orientation is counter-clockwise. Therefore, we obtain a cross-cap drawing for a scheme where the permutation around the right vertex is $I$ and around the left vertex is the mirror of $\pi$, as required.

 In the remainder of this article, we often omit vertices $(-1,-n/2)$ and $(d+2,-n/2)$ in our illustrations as they play no role, and we always have that $\pi^{d+1}$ is $I$.  For clarity, in all our figures, we write the sc-permutations on the left and the right. When writing the sc-permutation $I$, even though all of the signs of its elements are positive, we sometimes add overlines indicating the signs of the corresponding elements in $\pi^1$: the point is that it allows the reader to see at a glance which edges in the embedding scheme are positive or negative. Also, the edges are not always drawn exactly x-monotone for cosmetic and readability purposes.

\section{The counterexample}\label{counter example}
In this section, we provide a family of 2-vertex embedding schemes that do not admit a perfect cross-cap drawing. Then we provide an explicit pseudo-triangulation of $N_5$ (depicted in Figure~\ref{G}), disproving Conjecture~\ref{main}.

\begin{figure}[t]
    \centering

     \begin{tikzpicture}[scale=0.25]
   
     \draw[line width=.05cm,red](-16,0).. controls(-12,10) and (-9,14.5) .. (-4,15) node [midway, sloped]{\textbf{$>$}}node[right] {\normalsize{$p^{'}$}};
     \draw[line width=.04cm,gray](-16,0).. controls(-12,8.5) and (-9,12) .. (-6,13)node [midway, sloped]{\textbf{$>$}}node[right] {\normalsize{$\bar{o}_1$}};
     \draw[line width=.05cm,red](-16,0).. controls(-12,6.5) and (-9,11.7) .. (-4,11)node [midway, sloped]{\textbf{$>$}}node[right] {\normalsize{$p_2$}};
     \draw[line width=.04cm,gray](-16,0).. controls(-12,4.5) and (-9,9) .. (-6,9)node [midway, sloped]{\textbf{$>$}}node[right] {\normalsize{$\bar{o}_2$}};
     \draw[line width=.05cm,red](-16,0).. controls(-12,3.5) and (-9,7.8) .. (-4,7)node [midway, sloped]{\textbf{$>$}}node[right] {\normalsize{$p_1$}};
     \draw[line width=.04cm,gray](-16,0).. controls(-12,2.5) and (-9,5) .. (-6,5)node [midway, sloped]{\textbf{$>$}}node[right] {\normalsize{$\bar{o}_3$}};
     \draw[line width=.05cm,red](-16,0).. controls(-12,1.5) and (-9,3) .. (-4,3)node [midway, sloped]{\textbf{$>$}}node[right] {\normalsize{$p^{''}$}};
     \draw[line width=.04cm,gray](-16,0).. controls(-10,1) .. (-6,1)node [midway, sloped]{\textbf{$>$}}node[right] {\normalsize{$\bar{b}_2$}};
     \draw[line width=.05cm,blue](-16,0).. controls(-10,-1) .. (-4,-1)node [midway, sloped]{\textbf{$>$}}node[right] {\normalsize{$\bar{n}^{'}$}};
     \draw[line width=.04cm,gray](-16,0).. controls(-12,-1.5) and (-9,-3) .. (-6,-3)node [midway, sloped]{\textbf{$>$}}node[right] {\normalsize{$\bar{e}_1$}};
     \draw[line width=.05cm,blue](-16,0).. controls(-12,-2.5) and (-9,-6) .. (-4,-5)node [midway, sloped]{\textbf{$>$}}node[right] {\normalsize{$\bar{n}_1$}};
     \draw[line width=.04cm,gray](-16,0).. controls(-12,-3.5) and (-9,-7) .. (-6,-7)node [midway, sloped]{\textbf{$>$}}node[right] {\normalsize{$\bar{e}_2$}};
     \draw[line width=.05cm,blue](-16,0).. controls(-12,-4.5) and (-9,-9.6) .. (-4,-9)node [midway, sloped]{\textbf{$>$}}node[right] {\normalsize{$\bar{n}_2$}};
     \draw[line width=.04cm,gray](-16,0).. controls(-12,-6.5) and (-9,-11) .. (-6,-11)node [midway, sloped]{\textbf{$>$}}node[right] {\normalsize{$\bar{e}_3$}};
     \draw[line width=.05cm,blue](-16,0).. controls(-12,-8.5) and (-9,-13.3) .. (-4,-13)node [midway, sloped]{\textbf{$>$}}node[right] {\normalsize{$\bar{n}^{''}$}}; 
     \draw[line width=.04cm,gray](-16,0).. controls(-12,-10) and (-9,-14.5) .. (-6,-15)node [midway, sloped]{\textbf{$>$}}node[right] {\normalsize{$\bar{b}_1$}};        
    
     \draw[line width=.05cm,red](16,0).. controls(12,10) and (9,14.5) .. (4,15)node [midway, sloped]{\textbf{$>$}}node[left] {\normalsize{$p^{'}$}};
     \draw[line width=.04cm,gray](16,0).. controls(12,8.5) and (9,12) .. (6,13)node [midway, sloped]{\textbf{$>$}}node[left] {\normalsize{$\bar{l}_1$}};
     \draw[line width=.05cm,red](16,0).. controls(12,6.5) and (9,11.7) .. (4,11)node [midway, sloped]{\textbf{$>$}}node[left] {\normalsize{$p_1$}};
     \draw[line width=.04cm,gray](16,0).. controls(12,4.5) and (9,9) .. (6,9)node [midway, sloped]{\textbf{$>$}}node[left] {\normalsize{$\bar{l}_2$}};
     \draw[line width=.05cm,red](16,0).. controls(12,3.5) and (9,7.8) .. (4,7)node [midway, sloped]{\textbf{$>$}}node[left] {\normalsize{$p_2$}};
     \draw[line width=.04cm,gray](16,0).. controls(12,2.5) and (9,5) .. (6,5)node [midway, sloped]{\textbf{$>$}}node[left] {\normalsize{$\bar{l}_3$}};
     \draw[line width=.05cm,red](16,0).. controls(12,1.5) and (9,3) .. (4,3)node [midway, sloped]{\textbf{$>$}}node[left] {\normalsize{$p^{''}$}};
     \draw[line width=.04cm,gray](16,0).. controls(10,1) .. (6,1)node [midway, sloped]{\textbf{$>$}}node[left] {\normalsize{$a_2$}};
     \draw[line width=.05cm,blue](16,0).. controls(10,-1) .. (4,-1)node [midway, sloped]{\textbf{$>$}}node[left] {\normalsize{$\bar{n}^{''}$}};
     \draw[line width=.04cm,gray](16,0).. controls(12,-1.5) and (9,-3) .. (6,-3)node [midway, sloped]{\textbf{$>$}}node[left] {\normalsize{$m_1$}};
     \draw[line width=.05cm,blue](16,0).. controls(12,-2.5) and (9,-6) .. (4,-5)node [midway, sloped]{\textbf{$>$}}node[left] {\normalsize{$\bar{n}_1$}};
     \draw[line width=.04cm,gray](16,0).. controls(12,-3.5) and (9,-7) .. (6,-7)node [midway, sloped]{\textbf{$>$}}node[left] {\normalsize{$m_2$}};
     \draw[line width=.05cm,blue](16,0).. controls(12,-4.5) and (9,-9.6) .. (4,-9)node [midway, sloped]{\textbf{$>$}}node[left] {\normalsize{$\bar{n}_2$}};
     \draw[line width=.04cm,gray](16,0).. controls(12,-6.5) and (9,-11) .. (6,-11)node [midway, sloped]{\textbf{$>$}}node[left] {\normalsize{$m_3$}};
     \draw[line width=.05cm,blue](16,0).. controls(12,-8.5) and (09,-13.3) .. (4,-13)node [midway, sloped]{\textbf{$>$}}node[left] {\normalsize{$\bar{n}^{'}$}}; 
     \draw[line width=.04cm,gray](16,0).. controls(12,-10) and (9,-14.5) .. (6,-15)node [midway, sloped]{\textbf{$>$}}node[left] {\normalsize{$\bar{a}_1$}};     
          
    \draw[line width=.04cm,gray] (-13,-22)-- (-13,-18.5)node [midway, sloped]{\textbf{$<$}}node[above] {\normalsize{$\bar{o}_3$}};  
    
    \draw[line width=.04cm,gray] (-13,-22)-- (-9.5,-20)node [midway, sloped]{\textbf{$<$}}node[right] {\normalsize{$\bar{l}_3$}};   
    \draw[line width=.04cm,gray] (-13,-22)-- (-16.5,-20)node [midway, sloped]{\textbf{$<$}}node[left] {\normalsize{$\bar{l}_2$}}; 
    
    \draw[line width=.04cm,gray] (-13,-22)-- (-9.5,-24.5)node [midway, sloped]{\textbf{$<$}}node[right] {\normalsize{$\bar{o}_2$}}; 
    \draw[line width=.04cm,gray] (-13,-22)-- (-16.5,-24.5)node [midway, sloped]{\textbf{$>$}}node[left] {\normalsize{$\bar{o}_1$}}; 
    
    \draw[line width=.04cm,gray] (-13,-22)-- (-13,-25.5)node [midway, sloped]{\textbf{$>$}}node[below] {\normalsize{$\bar{l}_1$}};

    \draw[line width=.04cm,gray] (0,-22)-- (0,-18.5)node [midway, sloped]{\textbf{$<$}}node[above] {\normalsize{$\bar{e}_3$}}; 
    
    \draw[line width=.04cm,gray] (0,-22)-- (3.5,-20)node [midway, sloped]{\textbf{$<$}}node[right] {\normalsize{$m_1$}};   
    \draw[line width=.04cm,gray] (0,-22)-- (-3.5,-20)node [midway, sloped]{\textbf{$<$}}node[left] {\normalsize{$m_2$}};

    \draw[line width=.04cm,gray] (0,-22)-- (-3.5,-24.5)node [midway, sloped]{\textbf{$>$}}node[left] {\normalsize{$\bar{e}_1$}}; 
    \draw[line width=.04cm,gray] (0,-22)-- (3.5,-24.5)node [midway, sloped]{\textbf{$<$}}node[right] {\normalsize{$\bar{e}_2$}};

    \draw[line width=.04cm,gray] (0,-22)-- (0,-25.5)node [midway, sloped]{\textbf{$>$}}node[below] {\normalsize{$m_3$}};

    \draw[line width=.04cm,gray] (13,-21.5)-- (13,-24.5)node [midway, sloped]{\textbf{$>$}}node[below] {\normalsize{$\bar{a}_2$}};   
    \draw[line width=.04cm,gray] (13,-21.5)-- (13,-18.5)node [midway, sloped]{\textbf{$>$}}node[above] {\normalsize{$\bar{a}_1$}}; 
    \draw[line width=.04cm,gray] (13,-21.5)-- (16,-21.5)node [midway, sloped]{\textbf{$<$}}node[right] {\normalsize{$b_2$}}; 
    \draw[line width=.04cm,gray] (13,-21.5)-- (10,-21.5)node [midway, sloped]{\textbf{$>$}}node[left] {\normalsize{$b_1$}};    
                    
    \filldraw[black] (16,0) circle (5pt); 
         \filldraw[black] (-16,0) circle (5pt);
         \filldraw[black] (-13,-22) circle (5pt);
         \filldraw[black] (0,-22) circle (5pt);
         \filldraw[black] (13,-21.5) circle (5pt);
   
     \end{tikzpicture}
     
    \caption{A pseudo-triangulation of $N_5$ admitting no perfect cross-cap drawing. The edges colored in blue and red depict the $2$-vertex embedding scheme described in Theorem~\ref{f}.}
    \label{G}
\end{figure}

\begin{remark}\label{closed}
    If a loopless $2$-vertex embedding scheme $G$ has one positive and one negative block, then so does any equivalent scheme under flips, therefore we do not need to account for the possible flips in the proof of Theorem~\ref{f}. 
\end{remark}

In order to prove Theorem~\ref{f}, we rely on Remark~\ref{closed} and Lemma~\ref{genus}.

\begin{restatable}
{lemma}{genus}\label{genus}
    Let $G$ be a loopless $2$-vertex embedding scheme that consists of a non-trivial positive block $A$ and a non-trivial negative block $B$. Let us denote by $G_A$ and $G_B$ the subschemes of $G$ that contain the edges in $A$ and $B$, respectively. Then $g(G)=g(G_A)+g(G_B)-1$.
\end{restatable}

The proof follows directly from the Euler characteristic. 

\begin{proof}
Assume that the positive block has edges labelled $A=\{e_1,\ldots, e_k\}$ and the edges of the negative block are $B=\{\overline{e}_{k+1}, \overline{e}_{k+2}, \ldots \overline{e}_{k+l}\}$. Notice that $f(G)=f(G_A)+f(G_B)-1$, indeed the face $e_1,e_k$ and the face $\overline{e}_{k+1},\overline{e}_{k+l}$ are the outer faces of $G_A$ and $G_B$, and they merge to become the face $e_1,\overline{e}_{k+l},e_k,\overline{e}_{k+1}$. Hence by Euler's formula $eg(G)=eg(G_A)+eg(G_B)+1$. Since the blocks $A$ and $B$ are non-trivial, $G_A$ and $G_B$ have non-zero Euler genus, and thus we have $g(G_A)=eg(G_A)+1$ and $g(G_B)=eg(G_B)+1$. On the other hand, a cycle in $G$ that contains one edge from $A$ and one edge from $B$ is one-sided, therefore $G$ is a non-orientable scheme, hence  $eg(G)=g(G)$.
All in all we can conclude 
\begin{align*}g(G)=eg(G)&= eg(G_A)+eg(G_B)+1 \\
&=g(G_A)-1+g(G_B)-1+1\\ &=g(G_A)+g(G_B)-1\end{align*} as claimed. 
\end{proof}

We now have all the tools to prove Theorem~\ref{f}, and refer to Figure~\ref{block example} for an example to help follow the proof.

\begin{proof}[Proof of Theorem~\ref{f}]

Let $G$ be a concatenation of a positive block $A$ with frames $p^{'}$ and $p^{''}$ and a negative block $B$ with frames $n^{'}$ and $n^{''}$. Let us assume that  $\phi$ is a perfect cross-cap drawing of $G$.
From Lemma~\ref{curves}
 we derive that $p^{'}\cdot n^{'}$ and $p^{'}\cdot n^{''}$ are orienting cycles, hence  by Lemma~\ref{separating} each of them enters each cross-cap once. 
Lemma~\ref{curves2} implies that $n^{'}\cdot n^{''}$ is separating. Therefore by Lemma~\ref{separating}, $n^{'}$ and $n^{''}$ enter the same cross-caps and do not enter any cross-cap that $p^{'}$ enters. Similarly,  $p^{'}\cdot p^{''}$ is separating and hence they enter the same cross-caps and no cross-cap that $n^{'}$ and $n^{''}$ enter. Then $A$ is drawn with at least $g(G_A)$ cross-caps and $B$ is drawn with at least $g(G_B)$ cross-caps that are disjoint from the cross-caps that $A$ entered. But by Lemma~\ref{genus} the non-orientable genus of $G$ is $g(G_A)+g(G_B)-1$. Therefore, there are not enough cross-caps available to draw both $A$ and $B$. This concludes.  
\end{proof}

Corollary~\ref{maint} follows at once as we can always add edges and vertices to a scheme to triangulate it without adding loops nor changing its genus, and any perfect cross-cap drawing of the triangulation restricts to a cross-cap drawing of the original embedding scheme. We provide in Figure~\ref{G} an example of such a pseudo-triangulation.

\section{Topology of the reversal distance}\label{reversal distance to genus}

In this section, we recall some well-known results from the genomics rearrangements literature\footnote{While the literature primarily deals with sorting standard signed permutations, here we are sorting cyclic signed permutations. 
In our description, following Section~\ref{cross-caps}, we directly translate their techniques to this cyclic setting, which involves some changes.}, which we interpret in the language of embedding schemes (see also~\cite{huang2017topological} for an alternate topological interpretation of these arguments).

\subsection{The Bafna-Pevzner inequality from Euler's formula}

Let $\pi$ be an sc-permutation, which, as explained in Section~\ref{preliminaries}, we think of as a 2-vertex embedding scheme, which requires $g(\pi)$ cross-caps to be drawn. We can compute its number of faces, which we denote by $f(\pi)$ and the number of elements in the sc-permutation corresponds to the number of edges in the scheme, which we denote by $e(\pi)$. Then Euler's formula reads 
$2-eg(\pi)=2-e(\pi)+f(\pi)$ which simplifies to 
$eg(\pi)=e(\pi)-f(\pi)$. By the construction in Section~\ref{cross-caps}, we thus have $d^c(\pi,id)\geq e(\pi)-f(\pi)$.

 A very similar inequality was first discovered by Bafna and Pevzner~\cite[Theorem~2]{bafna1996genome} without reference to embeddings. This inequality is sometimes strict, this is for instance the case for the embedding schemes of Theorem~\ref{counter example} and the example in Figure~\ref{counter example}.

The starting idea of the Hannenhalli--Pevzner (HP) algorithm is to identify intervals in an sc-permutation where applying an sc-reversal is clearly making progress. Given an sc-permutation $\pi$, we call a pair of consecutive integers $i$ and $i+1$ (modulo $n$) \emph{reversible} if they have opposite signs in $\pi$ (this is called an  \emph{oriented pair} in~\cite{bergeron2004reversal}). For a given reversible pair there exist two sc-reversals $\sigma, \sigma'$ such that $i$ and $i+1$ are homotopic in $\pi \cdot \sigma$ and in $\pi \cdot \sigma'$. These two reversals are equivalent in the sense that the two sc-permutations that they yield differ by a flip. For example consider the sc-permutation $(1,3,\bar{2},4)$. The integers $1$ and $2$ are reversible. To make $1$ and $2$ homotopic we can reverse the interval $(3,\bar{2})$ or the interval $(4,1)$. Applying these reversals we obtain $(1,2,\bar{3},4)$ and $(\bar{4},3,\bar{2},\bar{1})$, respectively, in which $1$ and $2$ are homotopic.

The following lemma follows from an Euler characteristic argument.

\begin{restatable}
 {lemma}{reversible}\label{reversible} If $(i,i+1)$ is a reversible pair in the sc-permutation $\pi$ and $\sigma$ is an sc-reversal that turns them into homotopic edges in $\pi \cdot \sigma$, then $eg(\pi\cdot \sigma)=eg(\pi)-1$.
   
\end{restatable}

\begin{proof}
Since $e(\pi)=e(\pi\cdot \sigma)$, we need to show that $f(\pi\cdot\sigma)=f(\pi)+1$. Since the edges $i$ and $i+1$ are consecutive in $I$, the edges $i$ and $i+1$ are incident to a face $a$ in $\pi$. The faces of $\pi \cdot \sigma$ can be computed by following the edges and switching sides depending on the signature. In doing so, one sees that the faces of $\pi\cdot \sigma$ are the same as the faces of $\pi$ except that the face $a$ is subdivided into two faces: one is the degree two face incident to $i$ and $i+1$ and the other face is incident to the other edges that were incident to $a$. No other face is disrupted by the sc-reversal, and thus this finishes the proof.
\end{proof}

\subsection{The HP algorithm for sc-permutations}

It is immediate to see that blocks do not contain reversible pairs of edges. More subtly, this is also the case for orientable portions, which generalize (minimal) blocks. These orientable portions will constitute the natural obstruction to applying the sc-reversals described above. Following the terminology for embedding schemes, we say that an sc-permutation is \emph{non-orientable} if it contains two elements of different signs, otherwise it is orientable. Similarly, it is \emph{reduced} if the associated embedding scheme is. \emph{Cleaning} an sc-permutation means identifying its homotopic elements, yielding a \emph{cleaned} permutation. 

Given a reversible pair $(i,i+1)$ in the sc-permutation $\pi$, let $\sigma$ be an sc-reversal that turns $i$ and $i+1$ into homotopic edges. The \emph{score} of $(i,i+1)$  is the number of reversible pairs in $\pi \cdot \sigma$. Note that the score of $(i,i+1)$ is independent of the choice of sc-reversal, since the two possible sc-reversals are equivalent. The basis of the HP algorithm, following the presentation of Bergeron~\cite{bergeron2001very} adapted to the cyclic case, is as follows.

\begin{framed}
\noindent \emph{HP algorithm:} 

\textbf{Input:} A reduced two-vertex non-orientable embedding scheme $G$.

\textbf{Output:} A perfect cross-cap drawing of $G$.

\textbf{Algorithm:} While there is a reversible pair, apply an sc-reversal on a pair of maximal score.
\end{framed}

Then the main theorem summarizing the properties of the HP algorithm is the following.

\begin{theorem}\label{HP}
If an sc-permutation $\pi$ is non-orientable and reduced, then $d^c(\pi,I)=eg(\pi)$, and the HP algorithm gives a sequence of sc-reversals of this optimal length.
\end{theorem}

This theorem follows at once from the following lemma. Recall that an sc-permutation is reduced if its component tree has only one connected component and to clean an sc-permutation we identify homotopic edges.

\begin{restatable}{lemma}{noblock}\label{noblock}
Let $\pi$ be a non-orientable and reduced sc-permutation, let $(i,i+1)$ be a reversible pair of maximal score, and let $\sigma$ be an sc-reversal that makes $i$ and $i+1$ homotopic, then either 
\begin{itemize}
\item $\pi \cdot \sigma$ is either $I$ or its flip, or 
\item $\pi \cdot \sigma$ is non-orientable and is reduced once it is cleaned.
\end{itemize}
\end{restatable}

We first show how to prove Theorem~\ref{HP} assuming the lemma, and postpone the proof of the lemma to the Section~\ref{sec:components_vs_portions}.

\begin{proof}[Proof of Theorem~\ref{HP}] The proof is by induction on $eg(\pi)$. Since by assumption $\pi$ is non-orientable, $eg(\pi)\geq 1$. 

For the initialization step, if $eg(\pi)=1$, by Lemma~\ref{reversible} a single sc-reversal transforms $\pi$ into an sc-permutation of Euler genus zero, i.e., the sc-permutation $I$ or its flip.

 For the induction step, first note that we always have $eg(\pi) \leq d^c(\pi,I)$, so we just need to prove the other inequality. Since $\pi$ is non-orientable and is reduced and $\sigma$ is a reversal of maximum score, by Lemma~\ref{reversible}, we have $eg(\pi \cdot \sigma)=eg(\pi)-1$, and by Lemma~\ref{noblock}, we have that $\pi\cdot \sigma$ is also non-orientable and is also reduced after cleaning it, i.e., identifying its homotopic edges. So we can use the induction hypothesis to obtain that $d^c(\pi\cdot \sigma, I)=eg(\pi\cdot \sigma)$. Now, since applying twice an sc-reversal undoes it, we have that $d^c(\pi,I) \leq d^c(\pi \cdot \sigma, I)+1$. We conclude by combining these three (in)equations and thus obtain a sequence of sc-reversals of optimal length after debundling the identified homotopic edges.
\end{proof}

\subsection{The interleaving graph}

The proof of Lemma~\ref{noblock} is inspired by the proof of~\cite[Theorem~1]{bergeron2001very}. There are two main differences in our setting, on one hand, the classical theory deals with standard signed permutations and we deal with signed cyclic permutations, adapting the ideas of \cite{bergeron2001very} to the cyclic setting is the content of this subsection. On the other hand, we need to translate between the portions tree and the interleaving graph. This is the content of the following subsection.

We will make heavy use of notation introduced in the preliminaries, which we recall: given an sc-permutation $\pi$, we use $(i,j)$ to denote the interval between $i$ and $j$ in $\pi$, and $[i,j]$ to denote the interval between $i$ and $j$ in $I$.

We first introduce an auxiliary graph which will be useful when analyzing the effect of sc-reversals. 

Given an sc-permutation $\pi$ on $n$ elements $\{1, \ldots, n\}$, we define its associated \emph{double permutation} $\pi^*=(\pi^*_1, \ldots \pi^*_{2n})$ to be a cyclic unsigned permutation on a set of $2n$ elements labeled $\{i^l, i^r \text{ for } 1\leq i\leq n\}$ in which $\pi^*_{2i-1}=\pi_i^l$ and $\pi^*_{2i}=\pi_i^r$ when $\pi_i$ is positive and $\pi^*_{2i-1}=\pi_i^r$ and $\pi^*_{2i}=\pi_i^l$ when $\pi_i$ is negative. See the examples in Figures ~\ref{fig:interleaving_graph}, ~\ref{fig:corners} and~\ref{fig:components_are_portions} .

Two intervals $(i^*,j^*)$ and $(k^*,\ell^*)$ in a cyclic permutation are called \emph{interleaving} if either $k^*$ belongs to $(i^*,j^*)$ and $\ell^*$ does not, or $\ell^*$ belongs to $(i^*,j^*)$ and $k^*$ does not.

We define the \emph{interleaving} graph $L_\pi$ of $\pi$ to have pairs of consecutive elements $(i,i+1)$ as vertices, and two vertices $(i,i+1)$ and $(j,j+1)$ are connected by an edge if the intervals $(i^r,i+1^l)$ and $(j^r,j+1^l)$ interleave in $\pi^*$.

We call a pair of edges $(i,j)$ a \emph{corner} if $i$ and $j$ are consecutive in the scheme induced by $\pi$, either around the vertex with sc-permutation $I$ or around the vertex with sc-permutation $\pi$. We can think of these corners topologically: such a corner is a small part of the surface in-between the two edges around the vertex, and thus belongs to some portion. Note that the corners in the vertex with the sc-permutation $I$ are of the form $(i,i+1)$, so they correspond to the vertex set of the interleaving graph $L_\pi$, and in what follows we refer to the vertices of the interleaving graph as \emph{corners}.

A corner $(i,i+1)$ is called \emph{reversible} if $(i,i+1)$ is a reversible pair in $\pi$, otherwise it is called \emph{non-reversible}. In our figures, we use white for reversible corners and black for non-reversible corners. See Figure~\ref{fig:interleaving_graph} for an example.

A connected component in $L_\pi$ is \emph{non-trivial} if it has more than one corner and it is \emph{orientable} if it only contains non-reversible corners. 
Notice that if $i$ and $i+1$ are homotopic in $\pi$, then the corner $(i,i+1)$ is an isolated vertex in the interleaving graph $L_\pi$. To shorten some statements, we use the notation $(i,i \pm 1)$ to refer to one of the corners $(i,i+1)$ or $(i-1,i)$. 

The interleaving graph serves as a bookkeeping device to track the effect of applying an sc-reversal turning a reversible pair into a pair of homotopic edges. This behavior is described in the following lemma (doing a reversal essentially switches to the complement on the neighborhood of a vertex). The neighborhood of the vertex $(i,i+1)$ is the induced subgraph that has as vertex set the corner $(i,i+1)$ itself and all the corners that are connected to $(i,i+1)$.
    
  \begin{lemma}\label{complement}
      If $(i,i+1)$ is a reversible pair in $\pi$ and $i$ is homotopic to $i+1$ in $\pi \cdot \sigma$, then the interleaving graph $L_{\pi\cdot \sigma}$ is identical to the interleaving graph $L_{\sigma}$ except in the closed neighborhood of $(i,i+1)$ in $L_\sigma$. In the closed neighborhood of $(i,i+1)$:\begin{itemize}
      \item a pair of corners is an edge in $L_{\pi \cdot \sigma}$ if it is not an edge in $L_\pi$, and
      
      \item vertices in $L_{\pi\cdot \sigma}$ are reversible if and only if they are non-reversible in $L_{\pi}$.
      \end{itemize}
  \end{lemma}

The proof of this lemma is identical to similar results in the signed reversal distance theory (see for example~\cite[Fact~2]{bergeron2001very}). 

\begin{proof}
  Let $(j,j+1)$ be a reversible corner connected to $(i,i+1)$ in $L_\pi$. This means that the intervals $(j^r,j+1^l)$ and $(i^r,i+1^l)$ interleave. One possible such case is when all these edges are positive, $j$ lies in $(i,i+1)$ and $j+1$ does not. Then if $\sigma$ is an sc-reversal that turns $i$ and $i+1$ into homotopic edges, it is immediate that $\sigma$ either switches the sign of $j$ while keeping $j+1$ untouched or switches the sign of $j+1$ while keeping $j$ untouched (depending on which of the two choices we made for $\sigma$). In both cases, the pair $(j,j+1)$ is no longer reversible. All the other cases follow identically by a straightforward case analysis.
  
  Similarly, let us look at one case where $(j,j+1)$ and $(k,k+1)$ are two neighbors of $(i,i+1)$ in $L_\pi$ and the intervals $(j^r,j+1^l)$ and $(k^r,k+1^l)$ interleave. This could be because all the corresponding edges are positive, $j$ and $k$ are in $(i,i+1)$, $j+1$ and $k+1$ are not and $j+1$ is in $(k,k+1)$ while $j$ is not. Then, after applying $\sigma$, either both $j$ and $j+1$ are in $(k,k+1)$ or none of them is (depending on which of the two choices we made for $\sigma$). Therefore, they stop being interleaved and the corners $(j,j+1)$ and $(k,k+1)$ are no longer connected in the interleaving graph. All other cases also follow by a straightforward case analysis -- in fact, the definition of doubled cyclic permutation $\pi^*$ was tailored for this analysis to work.
 \end{proof}

 Figure~\ref{fig:interleaving_graph} depicts the effect of applying the sc-reversal that turns the pair $(4,5)$ into homotopic edges. 

\begin{figure}
    \centering
    \includegraphics[width=.8\textwidth]{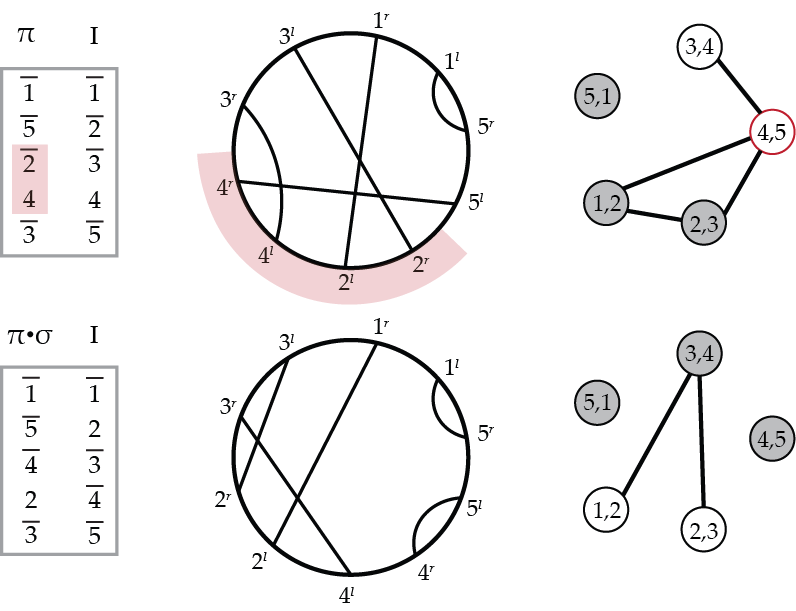}
    \caption{Top: the sc-permutation $\pi$, its associated doubled cyclic permutation and its interleaving graph $L_\pi$. At the bottom we can see the effect of applying a reversal $\sigma$ that makes $4$ and $5$ homotopic, by reversing elements $2,4$ on the interleaving graph (highlighted elements).}
    \label{fig:interleaving_graph}
\end{figure}

Recall that if $\sigma$ denotes the sc-reversal that turns $i$ and $i+1$ into homotopic edges, the score of $(i,i+1)$ is the number of reversible pairs in $\pi \cdot \sigma$. If we denote by $\#^+(i,i+1)$ (resp. $\#^-(i,i+1)$) the number of reversible (resp. non-reversible) pairs of edges adjacent to $(i,i+1)$ in $L_\pi$. If $\pi$ has $k$ reversible pairs of edges, then by Lemma~\ref{complement}, the reversible pairs adjacent to $(i,i+1)$ become non-reversible and vice-versa, and thus, \[\text{score}(i,i+1)=k-\#^+(i,i+1)+\#^-(i,i+1).\]

\begin{lemma}\label{key_interleving} If $\sigma$ is the sc-reversal that turns a reversible pair $(i,i+1)$ of highest score, into a pair $i$ and $i+1$ of homotopic edges, then the number of non-trivial orientable connected components in $L_{\pi\cdot\sigma}$ cannot be more than in $L_\pi$.

\end{lemma}

\begin{proof}
By Lemma~\ref{complement} $\sigma$ has no impact on corners not connected to $(i,i+1)$ in $L_\pi$.  We claim that if $U$ is a non-trivial orientable connected component that was created after applying $\sigma$, then the score of any pair $(j,j+1)$ in $U$ was higher than the score of the pair $(i,i+1)$.

 By Lemma~\ref{complement} all corners in $U$ were reversible and adjacent to $(i,i+1)$ before applying $\sigma$. Let $(j,j+1)$ be one such corner in $U$.  Before applying $\sigma$, every non-reversible corner that was connected to $(i,i+1)$ had to be connected to $(j,j+1)$. Otherwise after applying $\sigma$, by Lemma~\ref{complement} this corner would be a reversible corner connected to $(j,j+1)$ and therefore belonging to $U$ which is not possible because $U$ is orientable. This implies that $\#^-(i,i+1)\leq \#^-(j,j+1)$.  

Before applying $\sigma$, every reversible corner $(t,t+1)$ that was connected to $(j,j+1)$ had to be connected to $(i,i+1)$. Because, again Lemma~\ref{complement} implies that, if $(t,t+1)$ is not connected to $(i,i+1)$, after applying $\sigma$, this corner remains connected to $(j,j+1)$ without changing its reversibility. This means that $(t,t+1)$ is a reversible pair that belongs to $U$ which is once again not possible. This implies that $\#^+(j,j+1)\leq \#^+(i,i+1)$. But since $U$ is non-trivial, it has more than one corner, and a corner in $U$ to which $(j,j+1)$ is connected after the sc-reversal is a reversible corner that was connected to $(i,i+1)$ but not to $(j,j+1)$. Therefore, the last inequality is strict and we actually have $\#^+(j,j+1)< \#^+(i,i+1)$.

Combining everything, we have that
\begin{align*}
    \text{score}(i,i+1)&= k-\#^+(i,i+1)+\#^-(i,i+1) 
    \\
    &< k-\#^+(j,j+1)+\#^-(j,j+1)
    \\
    &=\text{score}(j,j+1)
\end{align*}

as we claimed. But this contradicts our assumption of $(i,i+1)$ having a maximal score, and hence contradicts the existence of $U$, and concludes the proof.

\end{proof}

\subsection{Relating components and portions}\label{sec:components_vs_portions}

The components of the interleaving graph $L_\pi$ are intimately connected to the portions of $\pi$. In this section, we concentrate on this connection which is finally used to prove Lemma~\ref{noblock}.

\begin{lemma}\label{componentcomponent}
    Let $\pi$ be a non-orientable sc-permutation for which $L_\pi$ has a non-trivial connected component $U$. Then the corners in $U$ are in the same non-trivial portion of $\pi$. 
\end{lemma}

The following proof is illustrated in Figure~\ref{fig:corners}.

\begin{proof}

    It is enough to show that if corners $(i,i+1)$ and $(j,j+1)$ are connected by an edge in $L_\pi$, then they belong to the same portion.
    The corner $(i,i+1)$ on the vertex corresponding to $I$ is in the same portion as the corner between $i$ and some edge $i'$ immediately adjacent to $i$ in the vertex corresponding to $\pi$ (after $i$ if $i$ is positive and before $i$ if $i$ is negative).
     Similarly, this defines the edges $(i+1)'$, $j'$ and $(j+1)'$. The two corners adjacent to $i$ correspond to the elements $i^l$ and $i^r$ in  $\pi^*$. The corner between $i$ and $i'$ is $i^r$, the corner between $i+1$ and $(i+1)'$ is $i+1^l$. Likewise for $j$.  
    
    For contradiction assume that some curve curve $k \cdot \ell$ separates the corners $(i,i+1)$ and $(j,j+1)$, then either $(i^r,i+1^l) \subseteq (k^r,\ell^l)$ and $(j^r,j+1^l)\subseteq (\ell^r,k^l)$ in $\pi^*$, or $(i^r,i+1^l) \subseteq (\ell^r, k^l)$ and $(j^r,j+1^l)\subseteq (\ell^r,k^l)$ in $\pi^*$. In both cases, we obtain a contradiction with the interleaving assumption between $(i,i+1)$ and $(j,j+1)$. Hence such a separating curve does not exist, and all the corners in $U$ belong to the same portion of $\pi$.    
\end{proof}

Notice that if we represent a double permutation by points on a circle, then $(i^*,j^*)$ and $(k^*,\ell^*)$ interleave if and only if the straight line\footnote{In our figures, the straight lines have often been made curvilinear for readability purposes.} segment between $i^*$ and $j^*$ crosses the straight line segment between $k^*$ and $j^*$. For briefness, we refer to the straight line segment between points $a^*$ and $b^*$ as the \emph{chord} between $a^*$ and $b^*$, and we denote it by $\langle a^*, b^*\rangle$. We will refer to this representation of $\pi^*$ as its \emph{circle representation}

For $U$ a set of vertices in $L_\pi$, define $U^*$ to be a subset of the elements of $\pi^*$, where $i^r \in U^*$ if $(i,i+1) \in U$ and $i^l \in U^*$ if $(i-1,i) \in U$.

\begin{lemma}\label{easylemma2}
 Let $a^*,b^*,c^*$ and $d^*$ be four distinct elements in $\{i^l, i^r \text{ for } 1\leq i\leq n\}$ such that $(a^*,b^*)$ and $(c^*,d^*)$ interleave in $\pi^*$. If $a^*$ and $b^*$ both belong to $U^*$ for $U$ a connected component of $L_\pi$, then there exists a corner $(k,k+1)$ in $U$ such that $(c^*,d^*)$ interleaves with $(k^r,k+1^l)$. 

\end{lemma}

\begin{figure}
    \centering
        \includegraphics[width=\textwidth]{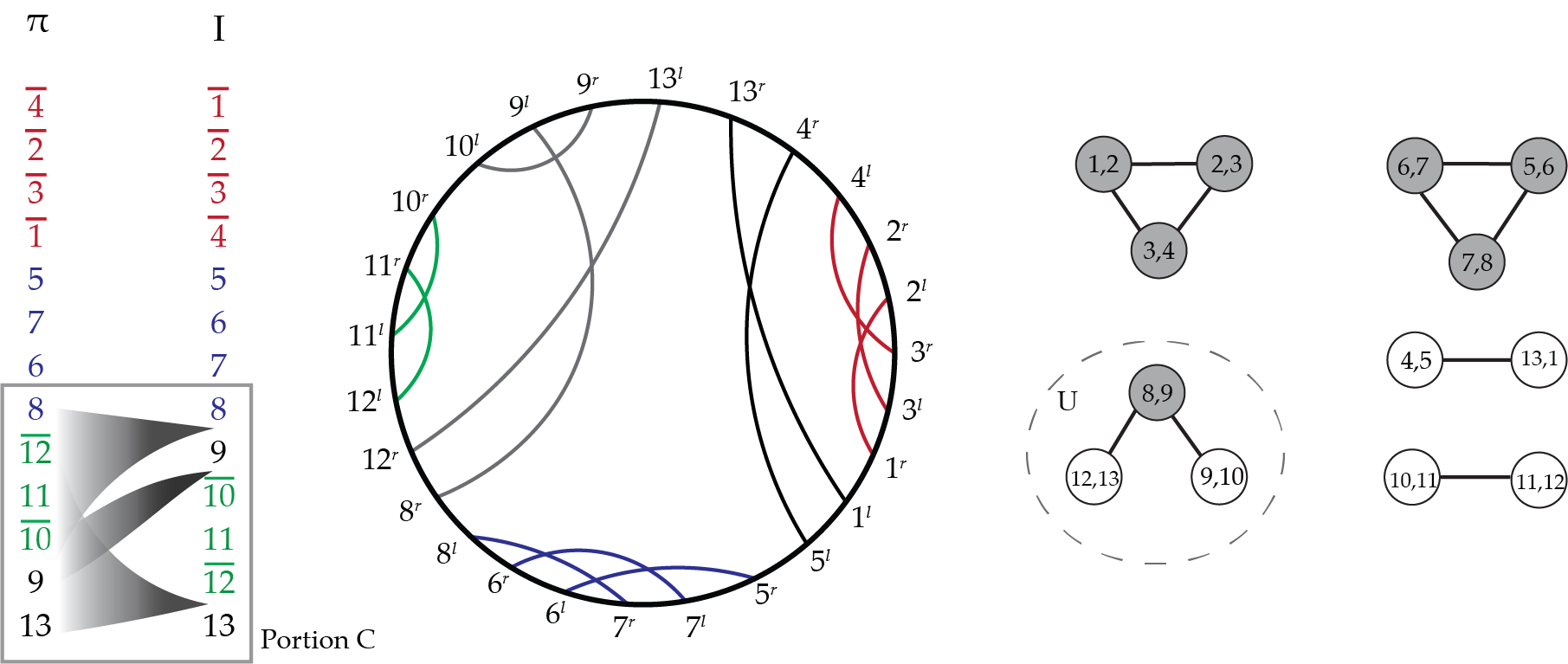}
    \caption{Picture for Lemma~\ref{componentcomponent}; the corners $(8,9)$,$(9,10)$ and $(12,13)$ belonging to the component $U$ are all in the portion $C$ in $\pi$.}
    \label{fig:corners}
\end{figure}

\begin{figure}
    \centering
        \includegraphics[width=\textwidth]{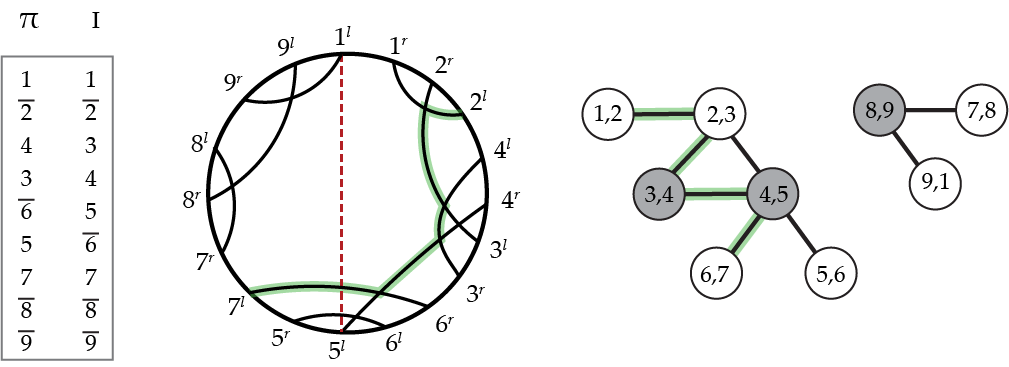}
    \caption{Picture for Lemma~\ref{easylemma2}; $(a^*,b^*)=(2^l,7^l)$ and $(c^*, d^*)=(5^l, 1^l)$.}
    \label{fig:components_are_portions}
\end{figure}

\begin{proof}
 In the circle representation of $\pi^*$ we add a chord between $\langle k^r,k+1^l\rangle$ for each $(k,k+1) \in U$. Recall that in this representation, the $\langle i^r,i+1^l\rangle$ and $\langle j^r,j+1^l\rangle$ intersect if and only if $(i^r,i+1^l)$ and $(j^r,j+1^l)$ interleave, i.e., if $(i,i+1)$ and $(j,j+1)$ are adjacent in $U$, see Figure~\ref{fig:components_are_portions}. 
  
  Since $a^*$ and $b^*$ are in $U^*$, there is a path $P$ in $U$ connecting $(a,a\pm 1)$ to $(b,b \pm 1)$. Consecutive vertices in $P$ correspond to intersecting chords in the circle representation, and thus the set of chords contains a piecewise-linear path $P^*$ between $a^*$ and $b^*$. Since $(c^*,d^*)$ interleaves with $(a^*,b^*)$, the chord $(c^*,d^*)$ would cross the piecewise-linear path $P^*$, and thus $(c^*,d^*)$ interleaves with some $(k^r,k+1^l)$ for $(k,k+1)$ a corner in $U$. This is illustrated in Figure~\ref{fig:components_are_portions}, where the red dotted line corresponding to $\langle c^*,d^*\rangle$ crosses the green piecewise-linear path.
\end{proof}

The next lemma relates the orientability of connected components of $L_\pi$ to the orientability of portions of $\pi$.

\begin{lemma}\label{orientabilitycomponent}
For a non-trivial orientable component $U$ of $L_\pi$, all the elements $i$ for $(i,i\pm1)$ in $U$ have the same signature.
\end{lemma}

\begin{proof}
   For $a^*$ and $b^*$ two elements of $\pi^*$, we denote by $\#(a^*,b^*)$ the number of elements in the interval $(a^*,b^*)$. Since $U$ is orientable, it does not contain corners $(i,i+1)$ where one element is positive and the other negative. A first observation follows from the definition of $\pi^*$: for a corner $(i,i+1)$ in $U$, the interval $(i^r,i+1^l)$ in $\pi^*$ contains an even number of elements of $\pi^*$. 
    
Let $i^*$ and $j^*$ denote two elements of $\pi^*$ in $U^*$ such that the interval $(i^*,j^*)$ in $\pi^*$ does not contain any other element in $U^*$ than its extremities. We say that such an interval is \emph{empty}. We claim that an empty interval $(i^*,j^*)$ also contains an even number of elements. Indeed, otherwise, by parity, some interval $(k^r,k+1^l)$ with one extremity strictly in $(i^*,j^*)$ interleaves with $(i^*,j^*)$, which by Lemma~\ref{easylemma2} implies that it interleaves with a pair $(a^r,a+1^l)$ for $(a,a+1)$ in $U$, and thus $k^r$ and $k+1^l$ belong to $U^*$. This contradicts the emptiness assumption.

We extend the $[\cdot]$ notation to the elements of $\pi^*$, denoting by $[a^*,b^*]$ the interval $(a^*,b^*)$ in the cyclic permutation $(1^l,1^r, 2^l, 2^r, \ldots , n^l,n^r)$. Now, similarly to the previous paragraph, we let $c^l$ and $d^r$ be two elements of $\pi^*$ in $U^*$ so that no other element of $[c^l,d^r]$ belongs to $U^*$, we call this interval a \emph{minimal interval}. Then for an element $k^*$ that is strictly in a minimal interval, $k^*$ does not belong to $U^*$, and $(k^r,k+1^l)$ does not interleave with $(c^l,d^r)$ (once again by Lemma~\ref{easylemma2}). So all the elements in a minimal interval $[c^l,d^r]$ belong to $(c^l,d^r)$ or $(d^r,c^l)$. In particular $c^r$ and $d^l$ are both in one of them. Since no interval of elements in $U^*$ can interleave with $(c^r, d^l)$ (still by Lemma~\ref{easylemma2}), this implies that the elements of $U^*$ are either all in $(c^l,d^r)$ or all in $(d^r,c^l)$, i.e., one of these two intervals is empty, and thus by the previous paragraph, $\#(c^l,d^r)$ is even.

Now, for any $a^r$ and $b^l$ in $U^*$, we have that $\#(a^r,b^l)=\sum_{i\in [a,b]}\#(i^r,i+1^l)$ modulo $2$, since for any $i$ in $[1,n]$, $\#(i^l,i^r)$ is always even. By the first observation, we can remove from this sum the summands $\#(i^r,i+1^l)$ for $(i,i+1)$ in $U$ since they are even. Grouping the remaining terms into extremities of minimal intervals and applying the conclusion of the previous paragraph shows that $\#(a^r,b^l)$ is even. It follows that $a$ and $b$ have the same sign, and thus all the elements $a$ for $(a,a\pm 1)$ in $U$ have the same sign.
\end{proof}

The following lemma helps us recognize pairs of edges that are separating in the sc-permutation from the circle representation of $\pi^*$.

\begin{lemma}\label{separating in cyclic} 
    If $a$ and $b$ have the same sign in $\pi$ and for any $i$ in $\pi$, the chord $\langle i^r,i+1^l \rangle $ crosses none of the chords $\langle a^r,b^l\rangle$ or $\langle b^r,a^l\rangle$,
    then $a \cdot b$ is a separating curve in $\pi$.
\end{lemma}

\begin{proof}

    Without loss of generality, we assume that $a$ and $b$ are positive.  Observe that if the chord $\langle i^r, i+1^l\rangle$ does not cross $\langle a^r, b^l\rangle $ and $\langle b^r,a^l\rangle$, then all the other elements in $\pi^*$ belong to either the interval $(a^r,b^l)$ or the interval $(b^r,a^l)$. For any element $f$ in $\pi$, $f^r$ and $f+1^l$ cannot belong to different intervals since otherwise the chord $\langle f^r,f+1^l \rangle$ will cross both $\langle a^r,b^l \rangle$ and $\langle b^r,a^l\rangle$. This implies that $(a^r,b^l)$ (resp. $(b^r,a^l)$) contains all the elements $j^*$ such that $j\in [a+1,b-1]$ (resp. $j\in [b+1,a-1]$). This means that $\omega_{a,b}$ encloses all the edges in the sc-permutation $\pi$ and therefore by Lemma~\ref{curves2}, $a\cdot b$ is separating. 
\end{proof}

The following lemma shows that a non-trivial orientable component in $L
_\pi$ corresponds to a non-trivial orientable portion in $\pi$. Since we do not merge the orientable portions, this gives a one-to-one correspondence between  orientable vertices in the portions tree and orientable components in the interleaving graph. 
\begin{lemma}\label{orientable components vs portions}
    If $\pi$ is an sc-permutation and $U$ is a non-trivial orientable component in $L_{\pi}$, then the corners in $U$ belong to a non-trivial orientable portion in $\pi$ and no other corner belongs to this portion. 
\end{lemma}
\begin{proof}

    Consider the double permutation $\pi^*$ and its circle representation. For any corner $(i,i+1)$ in $U$, there exists a corner $(j,j+1)$ in $U$ such that $(i^r, i+1^l)$ interleaves with $(j^r,j+1^l)$, equivalently, the chords $\langle i,i+1\rangle$ and $\langle j,j+1\rangle$ intersect in the circle representation.
    Furthermore, we know that no $(k,k+1)$ not in $U$ interleaves with the vertices in $U$.

By Lemma~\ref{orientabilitycomponent}, we know that all elements $i$ for $(i,i\pm1)$ in $U$ have the same signature; without loss of generality we assume that all the elements are positive. 

If $U$ is the only component in $L_\pi$ then every corner $(i,i+1)$  belongs to $U$. We can show that if there is a chord $\langle k^r,l^l\rangle$ that does not intersect any $\langle i^r,i+1^l\rangle$ for $i\in\pi$, then $\langle l^r,k^l\rangle$ intersects at least one. If not, the corners in $(k^r,l^l)$ and $(l^r,k^l)$ belong to two different components in the interleaving graph which is a contradiction.  
Therefore by Lemma~\ref{separating in cyclic}, there is no separating pairs of edges which means that $\pi$ has only one orientable portion.

If $U$ is not the only component, then there exists some element in $\pi^*$ that does not belong to $U^*$. Consider the maximal intervals $(i^*,j^*)$ in $\pi^*$ such that any $k^* \in (i^*,j^*)$ belongs to $U^*$. 
Denote by $F_n:=(i_n^*, j_n^*)$ for $1\leq n\leq p$ all such intervals, ordered consecutively for the order given by $\pi^*$, and denote by $E_n:=(j_n^*,i_{n+1}^*) \setminus \{j_n^*,i_{n+1}^*\}$ the complementary intervals. Figure~\ref{intervals of U} illustrates the following argument. 

\begin{claim}

For any $n$,  $i_n^*=i_n^r$ and $j_n^*=j_n^l$.
\end{claim}

\begin{figure}
    \centering
    \includegraphics[width=.45\textwidth]{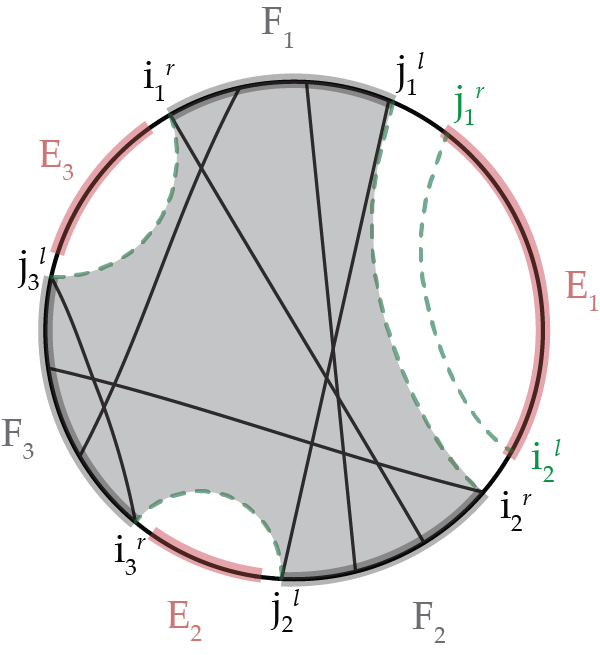}
    \caption{The shaded area depicts the region where the crossing of chords corresponding to the elements in an orientable component $U$ happen. The dotted green lines distinguish the separating pairs of edges.}
    \label{intervals of U}
\end{figure}
\begin{proof}[Proof of the claim]
For any $h^r$ belonging to the interval $E_n$, we claim that $h+1^l$ also belongs to $E_n$. Indeed, similarly to the proof of Lemma~\ref{easylemma2}, consider the path of segments $P$ connecting $j_n^*$ and $i_{n+1}^*$. We observe that if $h+1^l$ does not belong to $E_n$, this means that $\langle h^r,h+1^l\rangle$ crosses $P$ and therefore $(h,h+1)$ belongs to $U$ (Lemma~\ref{easylemma2}) which is a  contradiction. Therefore $h+1^l$ belongs to $E_n$. The argument is symmetric, and therefore $E_n$ is a disjoint union of elements of the form $(h^r,h+1^l)$. Since $E_n$ is not the entire $\pi^*$, it must therefore contain exactly one element $h^r$ for which $h^l$ is not in $E_n$: this must be $i_{n+1}^l$ because $i$ and $h$ have positive signature. Similarly, it also contains exactly one element $h^l$ for which $h^r$ is not in $E_n$, and this must be $j_n^r$ because $i$ and $h$ have positive signature. It follows that for any $n$, $i_n^l$ and $j_n^r$ are not in $F_n$, which proves the claim. \renewcommand\qedsymbol{$\lrcorner$}\end{proof}

This claim implies that for any $n$, the chord $\langle j_n^r, i_{n+1}^l\rangle$ does not cross any other chord corresponding to elements of $U$. Lemma~\ref{separating in cyclic} implies that the pairs $i_{n+1},j_n$ introduce separating curves that together separate all the elements $i$ for $(i,i\pm1)$ in $U$. Together with Lemma~\ref{componentcomponent}, this proves that the corners in $U$ belong to a non-trivial orientable portion in $\pi$, and no other corner belongs to this portion.
\end{proof}

A partial converse to Lemmas~\ref{componentcomponent} and~\ref{orientable components vs portions} is provided by the following lemma.

 \begin{lemma}\label{easy block}Let $\pi$ be a permutation, if either:
 \begin{itemize}
 \item $\pi$ contains a non-trivial orientable portion, or
\item $\pi$ is orientable and not equal to the sc-permutation $I$ or its flip,
\end{itemize}
 
 then the interleaving graph $L_\pi$ contains a non-trivial orientable connected component.

\end{lemma}

\begin{proof}
    We first consider the case of a non-trivial orientable portion where all the edges are positive and we denote its edges by $C$. Let $a$ and $b$ be minimal such that $C \subseteq (a,b)$, therefore $a$ and $b$ are in $C$. Since $a-1$ is not in $C$, $a+1$ is, and by non-triviality, $b \neq a+1$. Furthermore, $a+1$ is not homotopic to $a$ otherwise $(a,a+1)$ would be separating. Thus some other element $j$ of $C$ lies in the interval $(a,a+1)$. This implies that $(a^r,a+1^l)$ interleaves with $(j^r,b^l)$, and thus an immediate induction shows that there is another corner in the connected component $U$ of $(a,a+1)$. By Lemma~\ref{componentcomponent}, all the corners in $U$ are in $C$ and thus the elements $i$ for $(i,i\pm1)$ in $U$ are all positive. This proves the first statement.
    
      The case where the edges in $C$ are negative follows from the fact that flipping $\pi$ changes neither its interleaving graph nor the orientability of its connected components.

 For the case where $\pi$ is orientable and not the sc-permutation $I$ or its flip, the assumption implies that there are at least two interleaving intervals in $\pi^*$ forming a non-trivial connected component $U$, and the orientability of $\pi$ implies the orientability of $U$.
\end{proof}

\begin{proof}[Proof of Lemma~\ref{noblock}]
By assumption $\pi$ is non-orientable and reduced, i.e., its portions tree consists of a single vertex.  By Lemma~\ref{orientable components vs portions} there is no non-trivial orientable component in $L_\pi$. 
Thus, by Lemma~\ref{key_interleving} there is also no non-trivial orientable component in $L_{\pi \cdot \sigma}$.

Lemma~\ref{easy block} implies that if $\pi \cdot \sigma$ is orientable then it is equal to the sc-permutation $I$ or its flip, and if $\pi \cdot \sigma$ is non-orientable then it does not have a non-trivial orientable portion. The first case is the first case of the conclusion of Lemma~\ref{noblock}. The second case implies that the portions tree of $\pi \cdot \sigma$ has only one vertex and therefore it is reduced (recall from the introduction that after cutting the surface into subsurfaces along all edges that participate in a separating curve, we define the portions by merging non-orientable sub-surfaces that share a boundary), which is the second case of the conclusion of Lemma~\ref{noblock}.

\end{proof}

\section{Perfect drawings for $2$-vertex graph embeddings}\label{S:perfect}

In this section, we prove Theorem~\ref{main theorem} for the case of 2-vertex loopless embedding schemes. Actually, we will prove the following slightly more general result.

\begin{theorem}\label{main theorem2}
For any embedding scheme $G$ of a 2-vertex loopless graph $G$, at least one of the following is true:  
 \begin{enumerate}
 \item $G$ admits a perfect cross-cap drawing,
 \item $G$ is non-orientable and every non-orientable portion $C$ is of one of exactly two types: either (i) $C|$ is isomorphic to the left picture on Figure~\ref{reduced badly} and $C$ has degree $2$ in the portions tree or (ii) $C|$ is isomorphic to the right picture of Figure~\ref{reduced badly} and $C$ has degree $4$ in the portions tree.
  \end{enumerate}
\end{theorem}

The construction behind the proof of Theorem~\ref{main theorem2} works by induction over the portions. The workhorse of this construction is a blowing up operation which acts as the reverse of the decomposition into portions and allows us to glue perfect cross-cap drawings of portions together. A \emph{ribbon} is a small neighborhood around an edge. Assume that we have a perfect cross-cap drawing of a $2$-vertex loopless graph $G$. Let $e$ be an edge in $G$ that enters at least one cross-cap, let $X$ be another $2$-vertex loopless graph, which is furthermore assumed to be reduced, and let $e'$ be an edge of $X$ of the same signature as $e$. A  \emph{blow up} of $G$ with $X$ at $(e,e')$, which we denote by $G \oplus_{(e,e')} X$ is obtained from the cross-cap drawing of $G$ by replacing a cross-cap $\mathfrak{c}$ that $e$ enters with $k$ cross-caps for some odd number $k$, and by drawing $X$ in a ribbon around $e$, where $e'$ has been duplicated in a pair of edges $a$ and $b$ as depicted in Figure~\ref{blowing up}. More precisely, we can assume by applying a homeomorphism that $e$ is drawn as a horizontal arc crossing $\mathfrak{c}$ and possibly other cross-caps. Then we align $k$ cross caps vertically in place of $\mathfrak{c}$, draw $a$ and $b$ identically to $e$ away from the new cross caps, and make both $a$ and $b$ enter each new cross cap in reverse orders (without loss of generality $a$ goes down and $b$ goes up). All the other edges of $G$ outside of $X$ that were entering $\mathfrak{c}$ are now drawn in the same way as $a$ or $b$ depending on whether they were above or below $e$. The edges of $X$ are drawn in the region between $a$ and $b$, with the other (possibly zero) $g(X)-k$ cross-caps.

The following lemma, which will be proved in Section~\ref{blowupsection}, ensures that a blow-up can always be realized while preserving perfection.

\begin{restatable}{lemma}{blowup}\label{blowup}
    Let $G$ be a non-orientable loopless 2-vertex embedding scheme admitting a perfect cross-cap drawing, where an edge $e$ enters at least one cross-cap. Then for any reduced 2-vertex loopless graph $X$ with an edge $e'$ of the same signature as $e$, we can blow-up $G$ with $X$ at $(e,e')$ so that the resulting drawing is perfect.
\end{restatable}

We say that a cross-cap drawing is \emph{fantastic} if it is perfect and every edge enters at least one cross-cap. We say that an embedding scheme $G$ \emph{admits a fantastic cross-cap drawing} if there exists a cross-cap drawing of $G$ that is fantastic. We say that an embedding $G$ \emph{admits almost fantastic cross-cap drawings} if for any choice of edge $e$ of $G$, there exists a cross-cap drawing of $G$ where all the edges of $G$ except $e$ enter at least one cross-cap. Note that after a blow-up, every edge that was entering at least one cross-cap in $G$ still enters a cross-cap, and all the new edges enter at least one cross-cap. So a blow-up that preserves perfection also preserves fantasticness.

It follows from the definition that a blow up $G \oplus_{(e,e')} X$ is a cross-cap drawing of a graph in which $(a,b)$ is a separating curve, and cutting along it yields $G$ on one side and $X$ on the other side (with copies of $e$ and $e'$ respectively on the boundaries). Therefore, Lemma~\ref{blowup} shows that in order to obtain a perfect cross-cap drawing of a $2$-vertex loopless graph, it suffices to draw one of the portions (almost) fantastically and then we can blow up the other portions in the drawing.

\begin{figure}[H]
    \centering
    \includegraphics[width=.8\textwidth]{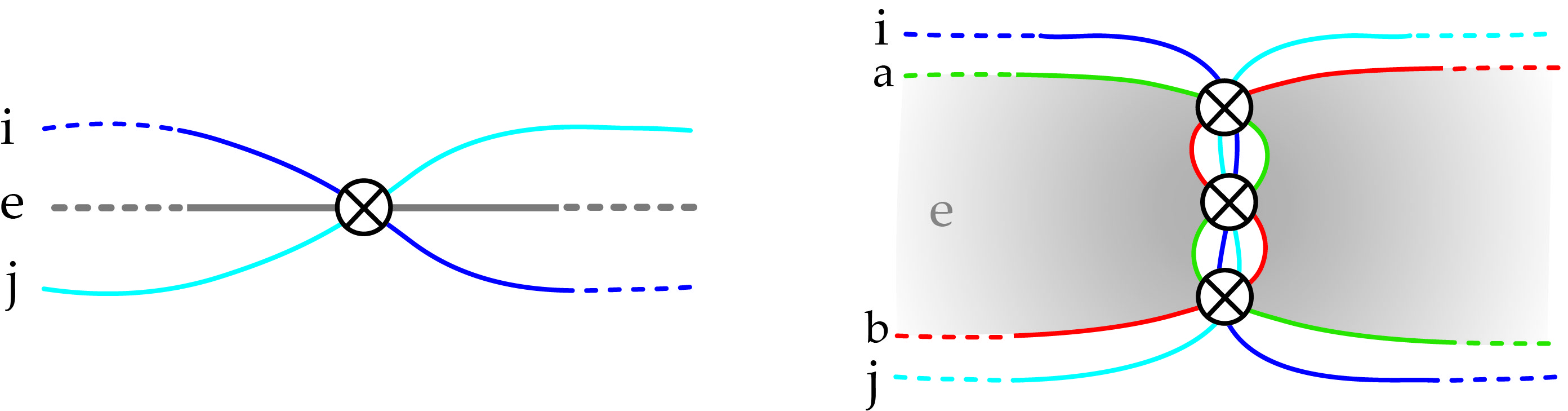}
    \caption{Blowing up a cross-cap.}
    \label{blowing up}
\end{figure}

The following lemma will be proved in the next subsection.

\begin{restatable}{lemma}{fantastique}\label{fantastique}
    Every non-orientable loopless 2-vertex embedding scheme that is reduced and different from $(1,\overline{2})$ or $(1,\overline{3},\overline{4},2)$ admits a fantastic cross-cap drawing.  Furthermore, both the embedding schemes $(1,\overline{2})$ and $(1,\overline{3},\overline{4},2)$ admit almost fantastic cross-cap drawings.
\end{restatable}

The following easy lemma shows that when trying to obtain perfect drawings, the orientable case can be reduced to the non-orientable one.

\begin{lemma}\label{adding edges}
    Let $G$ be a loopless $2$-vertex orientable embedding scheme with non-orientable genus $g(G)$. We can add an edge to $G$ to obtain a loopless $2$-vertex non-orientable embedding scheme $G'$ with the same non-orientable genus $g(G)$. If $G$ is reduced, so is $G'$.
\end{lemma}

\begin{proof}
    Let us assume that the signatures of the edges are positive. Choose an edge $f$.  Add a negative edge $e$ consecutive to $f$, such that in the cyclic permutation of one vertex we see $ef$ and in the other vertex we see $fe$. Let us denote the new embedding scheme by $G'$. The scheme $G'$ is non-orientable. Since $G$ is orientable, we know that $g(G)=eg(G)+1$. Since $e(G')=e(G)+1$, it is enough to show that $f(G')=f(G)$. By looking at the face cycles of both schemes, we can see that all face cycles are intact except a face $(i,f,j,\dots)$ in $G$ that is turned to a face cycle $(i,e,f,e,j,\dots)$ ($i$ and $j$ are two adjacent edges to $f$ in the cyclic permutation of the vertices). Since the edge $e$ has an opposite signature compared to the edges of $G$, $e$ cannot form a separating curve with another edge of $G$. Therefore, if $G$ is reduced, so is $G'$. The proof for the case where all the edges of $G$ are negative can be obtained by flipping. This finishes the proof. 
\end{proof}

We now have all the tools to prove Theorem~\ref{main theorem2}.

\begin{proof}[Proof of Theorem~\ref{main theorem2}]

If $G$ is orientable, without loss of generality we can assume that the signatures of all the edges are positive. Using Lemma~\ref{adding edges} we can add an edge with a negative signature to make it non-orientable. Since separating curves require two edges with the same signature, adding an edge of negative signature does not change the portions tree of $G$, and thus it changes one of the portions into a non-orientable portion. We claim that this non-orientable portion cannot be of one of the two types in the second case of the theorem. Indeed, it cannot be the scheme with permutation $\pi=(2,1,\bar{3},\bar{4})$ because it has only one edge of negative signature. It also cannot be the scheme with permutation $\pi=(1,\bar{2})$, because the negative signature edge is not involved in separating curves, and therefore we claim that the non-orientable portion must have at most one neighbor in the portions tree. Indeed, if there was a separating curve using (homotopic copies of) the edge $e$ of positive signature and it was adjacent to multiple connected components, in the portions tree there would be a single edge to the equivalence class of these edges. So we can assume that $G$ is non-orientable and the second case of the theorem does not hold.

If the second case of the theorem does not hold, then $G$ contains at least one non-orientable portion $C$ that is neither of type (i) nor (ii). If $C|$ is not isomorphic to $(1,\overline{2})$ or $(1,\overline{3},\overline{4},2)$, then it admits a fantastic drawing by Lemma~\ref{fantastique}. If it is isomorphic to one of them, then by the assumption on the degrees, one of its edges $e_0$ is such that no edge $e$ in the homotopy class of $e_0$ is a boundary edge in $C$.  In that case we use the second part of Lemma~\ref{fantastique} to obtain a perfect cross-cap drawing of $C|$ where all the edges except $e_0$ enter at least one cross-cap. For the sake of simplicity, we just say that we have drawn $C|$ almost fantastically in the description of our drawing algorithm below.

\begin{framed}
\emph{Main Algorithm}

\textbf{Input:} A two-vertex loopless embedding scheme $G$ with at least one non-orientable portion that is neither of type (i) nor (ii)

\textbf{Output:} {A perfect cross-cap drawing of $G$.}

\textbf{Initialization:} 
Compute the portions tree, and find the portion $C$ which is neither of type (i) nor (ii). Root the portions tree $T$ at $C$. Draw $C|$ almost fantastically.

\textbf{Algorithm:}
Do a breadth-first search traversal of $T$. Whenever a new portion $X$ is met, with an edge $e'$ on its boundary that has already been drawn as $e$, add $X$ to the cross-cap drawing by blowing up the current drawing with $X$ at $(e,e')$. 
\end{framed}

By construction, the edge $e_0$ will never be blown up. Thus Lemma~\ref{blowup} shows that each blow-up can be realized while preserving perfection. By construction, after each blow-up, the resulting drawing has the property that every edge enters at least one cross-cap, except possibly for the edge $e_0$. Therefore the algorithm provides a perfect cross-cap drawing of $G$, which concludes the proof.
\end{proof}

\subsection{Fantastic drawings}

We now prove Lemma~\ref{fantastique}, which we restate for convenience.

\fantastique*

The proof is based on exhaustive analysis of all the loopless $2$-vertex embedding schemes of non-orientable genus $2$ and $3$, as pictured in Figure~\ref{cases}. 

\begin{proof}[Proof of Lemma~\ref{fantastique}]
We first prove the first point. Let $\pi$ be the sc-permutation associated to the embedding scheme. 
Since $\pi$ is reduced and non-orientable, by Theorem~\ref{HP}, the HP algorithm provides a path of sc-permutations $\{\pi=\pi^1,\pi^2,\ldots \pi^{g+1}=id\}$, where $g$ is the non-orientable genus of $\pi$. We distinguish cases depending on the value of $g$.

If $g \geq 3$, we consider the sub-path $\{\pi^1,\pi^2,\ldots,\pi^{k}\}$, with $g-k=2$ and realize this sub-path as a cross-cap drawing $\phi$ as described in Section~\ref{preliminaries}. Notice that if there exists a fantastic cross-cap drawing $\phi'$ for $\pi^{k}$, we can concatenate $\phi$ with $\phi'$ to obtain a fantastic cross-cap drawing for $\pi$. Now, $\pi^{k}$ is an sc-permutation of non-orientable genus $3$. If $\pi^k$ contains homotopic edges, for each collection of homotopic edges we reduce them, which amounts to removing all of them but one. The proof then proceeds via an exhaustive case analysis.
 Without loss of generality, we can assume that 
\begin{itemize}
        \item The scheme $\pi^k$ is reduced since  this is preserved by the HP algorithm (this is a consequence of Lemma~\ref{noblock}).
        
        \item There is at least one edge of signature $-1$ and one edge of signature $+1$ in $\pi^k$. Otherwise, $\pi^k$ is orientable, which is impossible since by Lemma~\ref{noblock}, the HP algorithm preserves non-orientability.
    \item $\pi^k$ has a maximal number of edges while preserving these properties. Indeed, otherwise, we can add edges, draw the resulting scheme and remove these superfluous edges at the end.
\end{itemize}

With these simplifying assumptions at our disposal, we can exhaustively enumerate all the genus $3$ embedding schemes that match these. The numbers are small enough that this can be done by hand, and we also ran a computer search for safety. All these schemes have all their faces of degree $4$, and thus have six edges. Then, there are exactly eight loopless $2$-vertex embedding schemes matching our assumptions, their sc-permutations are:   
    $(1,\overline{6},5,\overline{4},3,\overline{2})$, $(1,\overline{6},\overline{3},5,4,\overline{2})$, $(1,\overline{3},\overline{6},5,\overline{4},2)$, $(1,5,\overline{3},\overline{6},4,2)$, $(1,\overline{4},6,3,\overline{5},2)$, $(1,\overline{6},3,\overline{4},5,\overline{2})$,
$(1,\overline{4},6,\overline{2},5,\overline{3})$, and $(1,\overline{4},\overline{6},\overline{2},5,3)$.  Fantastic drawings of each of them are provided in Figure~\ref{cases}.

If $g = 2$, there are exactly three reduced schemes: $(1,\overline{3},\overline{4},2)$, $(1,\overline{2},3,\overline{4})$ and $(1,3,\overline{2})$. Fantastic drawings of the second and third case (or rather its flipped version) are provided in Figure~\ref{cases}. 

If $g=1$, there is a single reduced scheme: $(1,\overline{2})$.

If $g=0$, there is a single reduced scheme: $(1)$.

Finally, homotopic edges that were removed from $\pi^k$ can be added back and drawn identically to their reduced edge. The resulting drawing is fantastic if and only if the reduced drawing is. This concludes the proof of the first point.

In order to prove the second point, it suffices to provide perfect cross-cap drawings of $(1,\overline{2})$ and $(1,\overline{3},\overline{4},2)$ where all edges except $e$ enter at least one cross-cap for every possible choice of edge $e$. This is done in Figure~\ref{almostfantastic}.
\end{proof}

\begin{figure}
    \centering
        \includegraphics[width=.9\textwidth]{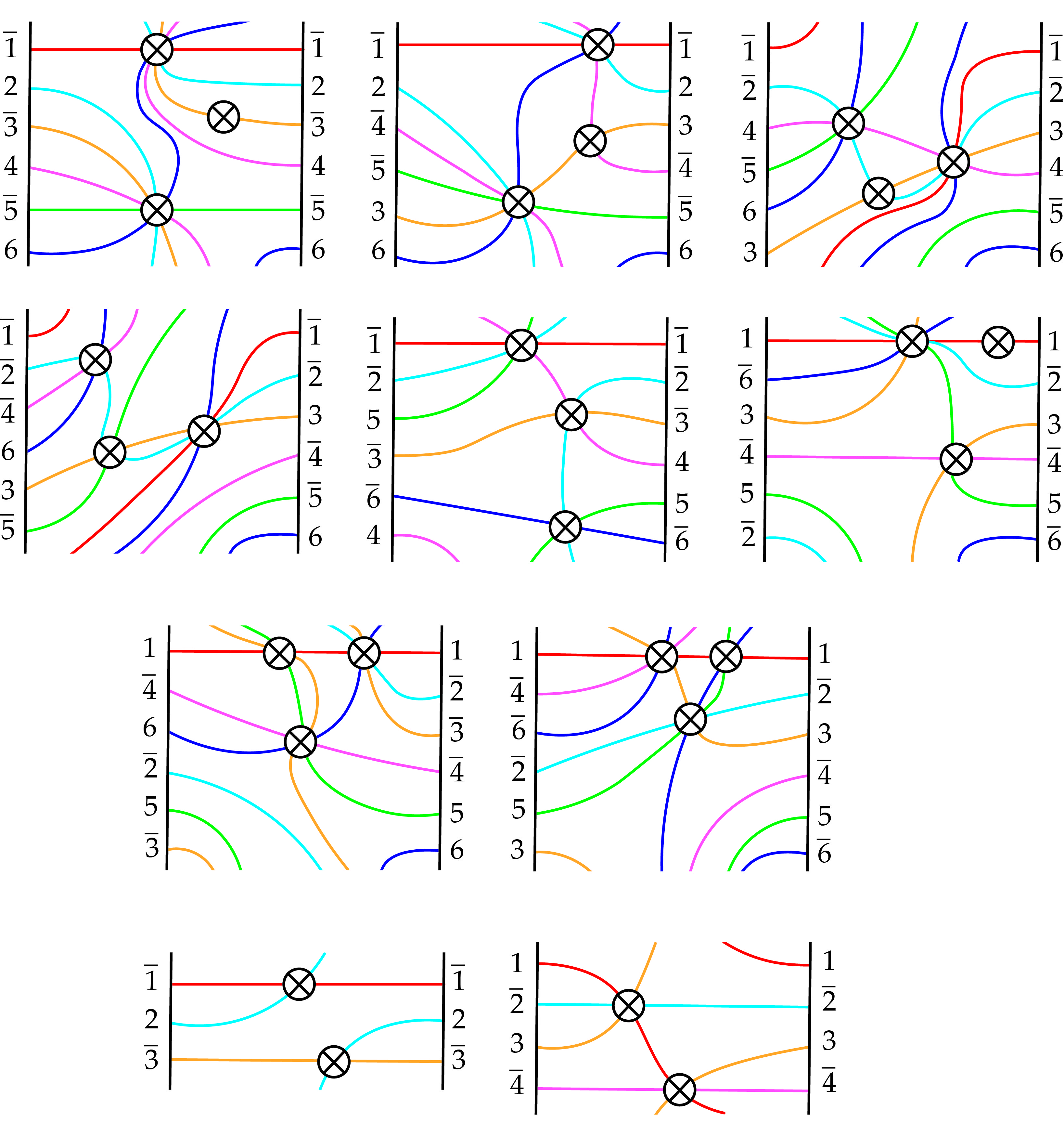}
    \caption{Fantastic drawings of non-orientable genus-$2$ and genus-$3$ embedding schemes (these are cylindrical drawings: the top is identified to the bottom).}
    \label{cases}
\end{figure}
\begin{figure}
    \centering
        \includegraphics[width=.6\textwidth]{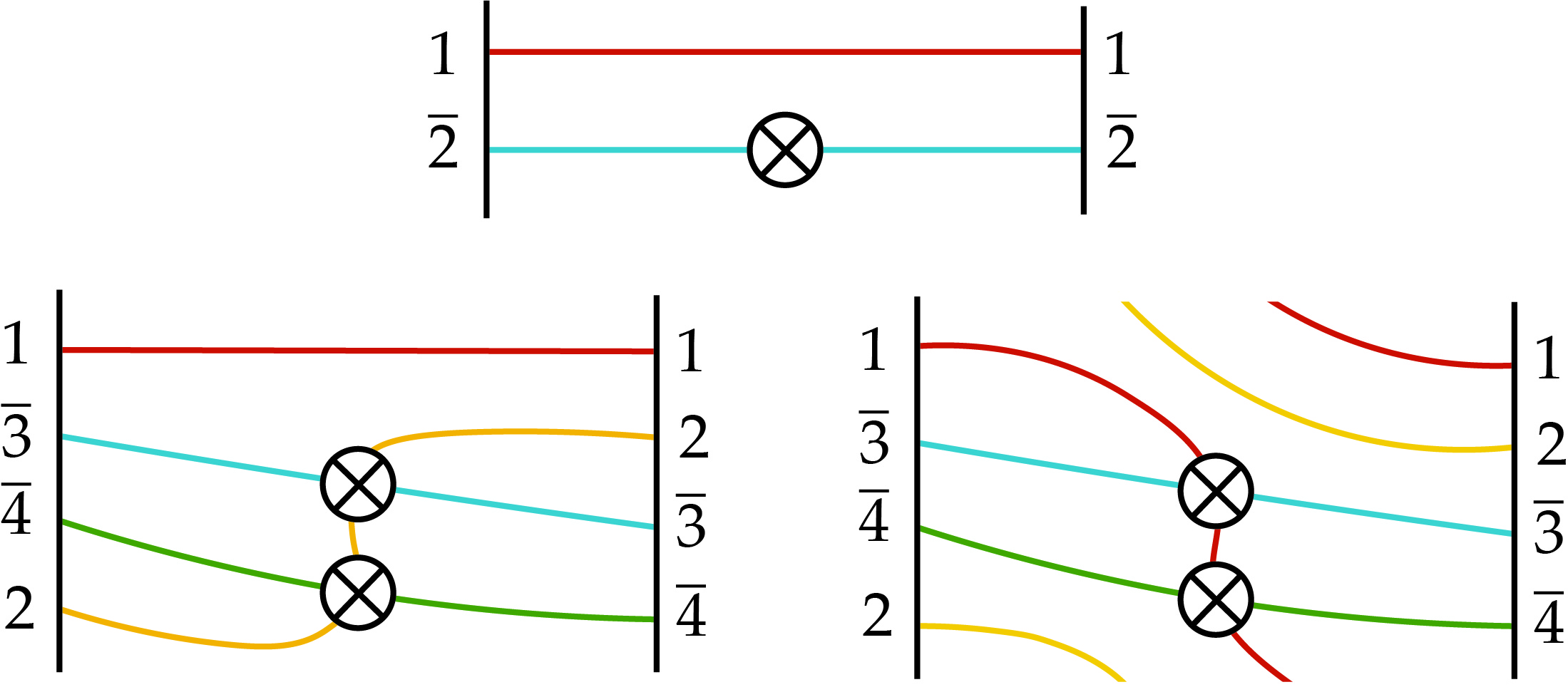}
    \caption{Almost fantastic drawings.
    Top: $(1,\overline{2})$, here $1$ does not enter any cross-cap. Bottom: $(1,\overline{3},\overline{4},2)$, at left (resp. right) the edge $1$ (resp. $2$) does not enter any cross-cap. The rest of the examples can be obtained flipping these ones.}
    \label{almostfantastic}
\end{figure}

\subsection{Drawing in ribbons and blowing up cross-caps.}\label{blowupsection}

We now prove Lemma~\ref{blowup} which we restate for convenience.

\blowup*

The key construction allowing us to blow up is given by the following lemma.

\begin{lemma}\label{blowupcorrect}
    Let $\pi$ be a reduced orientable sc-permutation on $n$ elements where all the elements have negative signature, and let $i$ be an element of $\pi$. Let $\pi'$ be the sc-permutation where the element corresponding to $i$ has been duplicated into two homotopic elements $i_1$ and $i_2$, and let $G'$ be the associated embedding scheme. Then $G'$ admits a perfect cross-cap drawing in which $i_1$ and $i_2$ each enter all the cross-caps in opposite order.
\end{lemma}

This lemma shows that $i_1$ and $i_2$ can be drawn exactly as specified by the blow-up operation.

\begin{proof}
Let us denote by $g(\pi')=eg(\pi')+1$ the non-orientable genus of $G'$. By renumbering, we can assume that $\pi'$ can be written in cycle notation as $(\pi'_1, \ldots , \pi'_n)$ with $\pi'_1=\overline{n}$ and $\pi'_n=\overline{1}$. Let us define $\pi''$ from $\pi'$ by replacing $\overline{1}$ by $n+1$, yielding a new sc-permutation on the elements $\{2, \ldots ,n+1\}$. The following lemma is proved using the HP algorithm and Theorem~\ref{HP}.

\begin{restatable}{lemma}{claimlemma}\label{claim}
The optimal number of sc-reversals to go from the sc-permutation $\pi''$ to the sc-permutation $(2,3,4,\dots, n+1)$ is $g(\pi'')=g(\pi')$. There exists a sequence of such sc-reversals such that no sc-reversal is applied to the element $n+1$.
\end{restatable}
\begin{proof}
Note that the associated embedding scheme to $\pi''$ is a non-orientable scheme and therefore $g(\pi'')=eg(\pi'')$. The number of edges in $\pi'$ and $\pi''$ are equal. Therefore, in order to show that $eg(\pi'')=g(\pi')=eg(\pi')+1$, it is enough to show that $f(\pi'')=f(\pi')-1$. Let us denote the elements of $\pi'$ by $\pi'=(\overline{n},\dots, \overline{2},\dots,\overline{j},\overline{1})$. Then the elements of $\pi''$ are $(\overline{n}, \dots,\overline{2},\dots,\overline{j},n+1)$. We see that $\pi'$ has at least two faces: $f_1=(1,n)$ and $f_2=(2,1,j, \dots)$. Replacing $-1$ with $n+1$, the faces $f_1$ and $f_2$ merge to a single face $f=(2,n+1,n,n+1,j,\dots)$ in $\pi''$ (see Figure~\ref{a2}). The other faces in $\pi''$ are the same as the faces in $\pi'$ other than $f_1$ and $f_2$. This implies that $\pi''$ has one face less than $\pi'$. 
\begin{figure}[H]
    \centering
    \includegraphics[width=.45\textwidth]{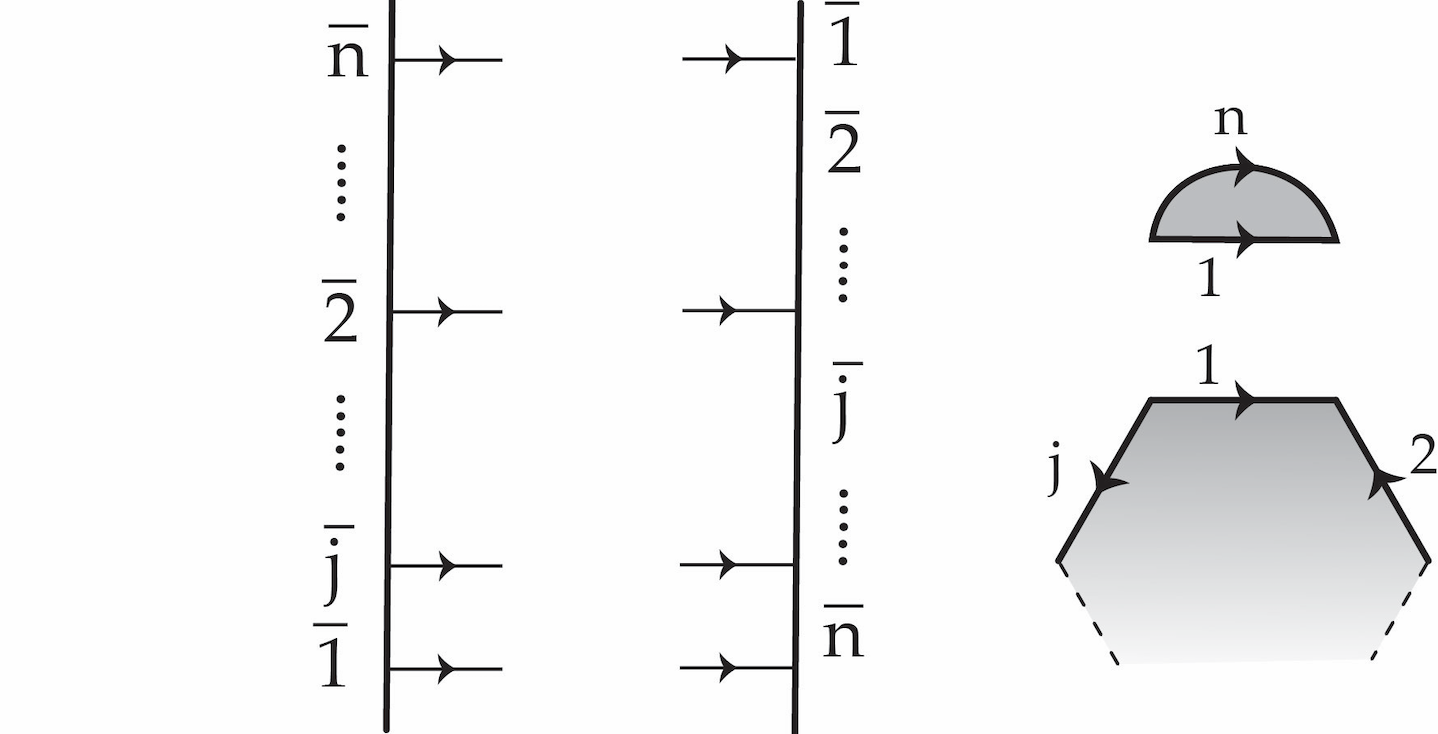}\hspace{1.2cm}\includegraphics[width=.45\textwidth]{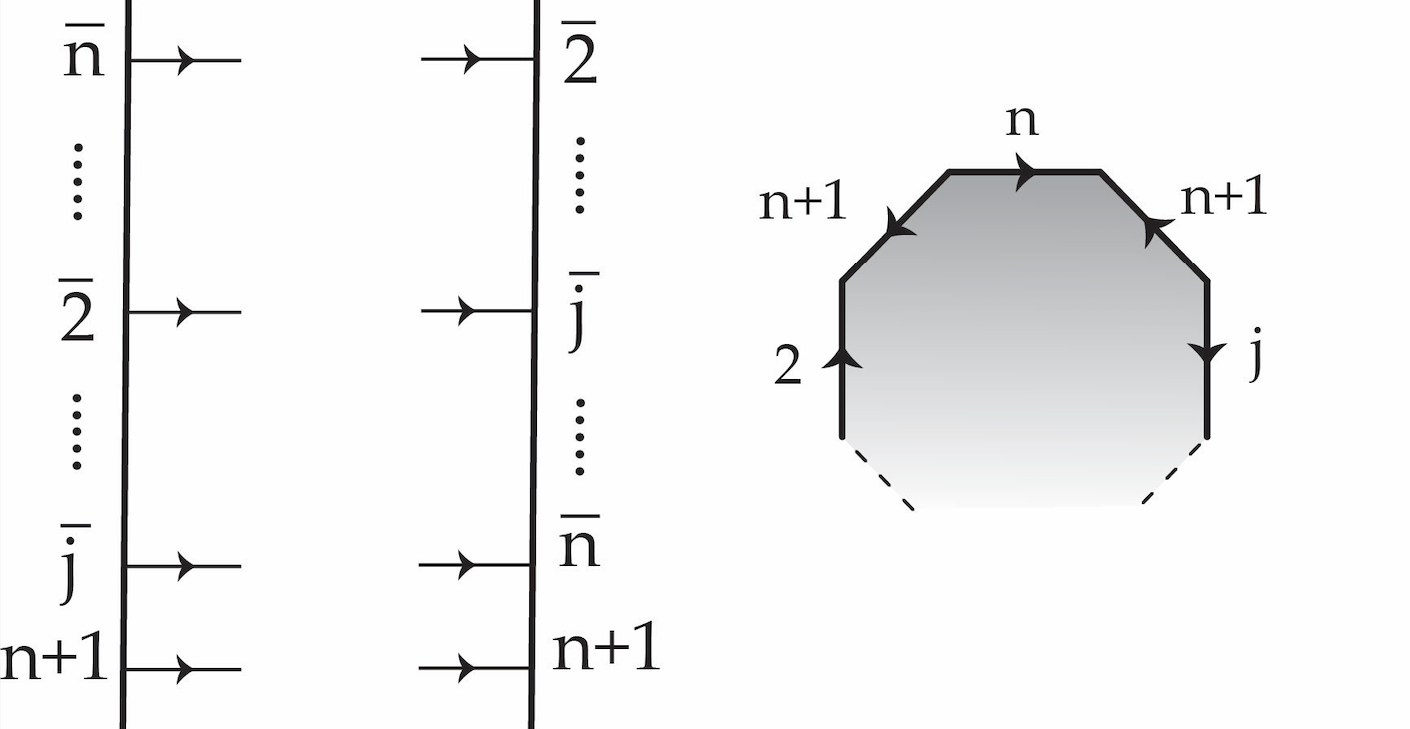}
   
    \caption{Faces of $\pi'$ (left) and $\pi''$ (right).}
    \label{a2}
\end{figure}

Furthermore, $\pi''$ is reduced because $\pi'$ is reduced. It is non-orientable by construction. Finally, recall that for a reversible pair $(i,i+1)$, there are two sc-reversals that make $i$ and $i+1$ homotopic. By running the HP-algorithm on $\pi''$ and choosing the sc-reversal at each step that does not reverse $n+1$, we obtain the desired sequence of sc-reversals. The success of the HP-algorithm is guaranteed by Theorem~\ref{HP}. This finishes the proof of the lemma.
\end{proof}

Now, the sequence of sc-reversals of Lemma~\ref{claim} gives us a cross-cap drawing for $\pi''$. By Lemma~\ref{curves}, the edges $n$ and $n+1$ form an orienting cycle, and thus together they have to enter all the $g(\pi'')$ cross-caps exactly once. We know that the edge $n+1$ does not enter any cross-cap in this drawing which implies that $n$ enters all the cross-caps exactly once. We can obtain a cross-cap drawing for $\pi'$ from this drawing for $\pi''$ by drawing $1$ entering the cross-caps that $n$ enters with the opposite order as depicted in Figure~\ref{a1}. 
\end{proof}
\begin{figure}[H]
    \centering
    \includegraphics[width=0.8\textwidth]{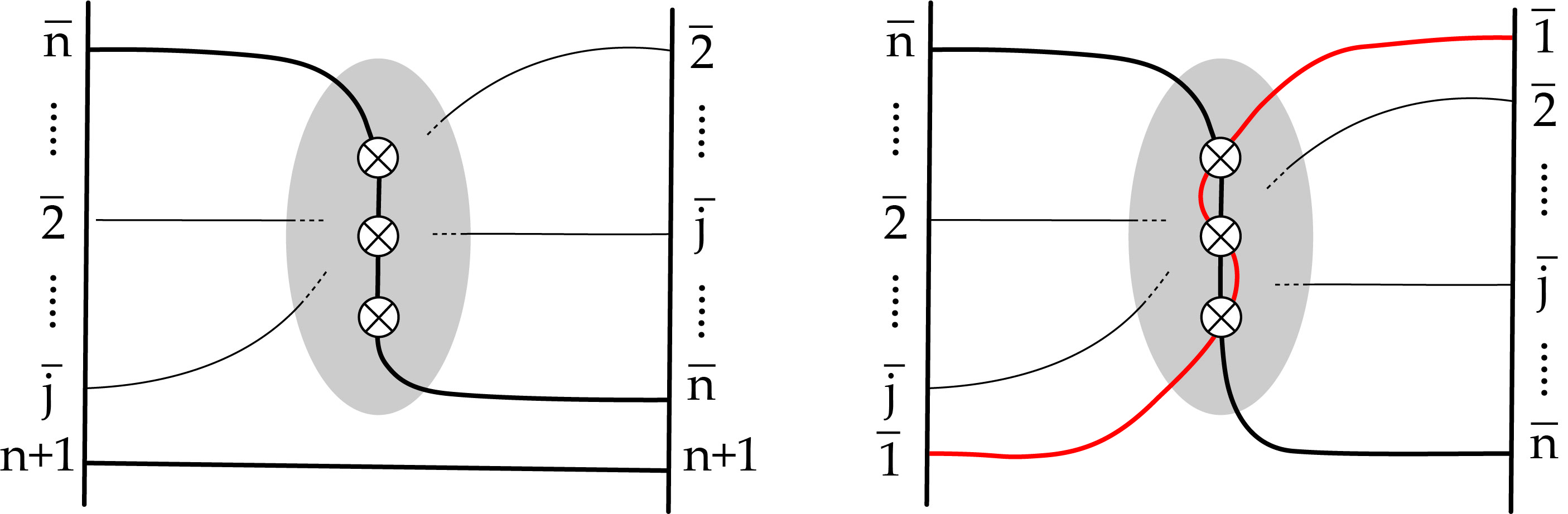}
    \caption{From a cross-cap drawing of $\pi'$ to a cross-cap drawing of $\pi$.}
    \label{a1}
\end{figure}

The proof of Lemma~\ref{blowup} can now be deduced from this lemma by applying (yet another time) the HP algorithm.

\begin{proof}[Proof of Lemma~\ref{blowup}]
If $X$ is orientable and its elements have negative signature, we can do the blow-up by replacing the first cross-cap that $e$ enters with all the $g(X)$ cross-caps required by $X$ in a vertical sequence. Then Lemma~\ref{blowupcorrect} shows that $X$ can be drawn perfectly inbetween $a$ and $b$.

If $X$ is orientable and its elements have positive signature, then $e$ also has positive signature and thus enters at least two cross-caps. Then we do the blow-up by replacing the second cross-cap that $e$ enters with $g(X)$ cross-caps and drawing $X$ having its edges first enter the first cross-cap used by $e$ and then applying Lemma~\ref{blowupcorrect} on the flip of $X$.

If $X$ is non-orientable, we can similarly assume that $e$ and $e'$ have negative signature. Then we apply the HP algorithm on $X$ to obtain an optimal sequence of sc-reversals, as promised by Theorem~\ref{HP}. Note that, as in the proof of Lemma~\ref{claim}, such a sequence of sc-reversals can be chosen so that $e'$ never enters a cross-cap. Since $e'$ has negative signature, $X$ has been transformed into the flip of $I$ by the HP algorithm. We can thus realize the blow-up by having all the edges of $X$ enter the first cross-cap entered by $e$ and replacing $e'$ by two copies $a$ and $b$ around these edges, see Figure~\ref{lastfigure}.  Note that since $X$ is non-orientable, the genus of $G \oplus_{(e,e')} X$ is $g(G)+g(X)$, which is the number of cross-caps used in this construction. Since the HP algorithm yields perfect drawings for non-orientable reduced sc-permutations (Theorem~\ref{HP}) and this construction preserves perfection, this concludes the proof.
\end{proof}
\begin{figure}[H]
    \centering
    \includegraphics[width=\textwidth]{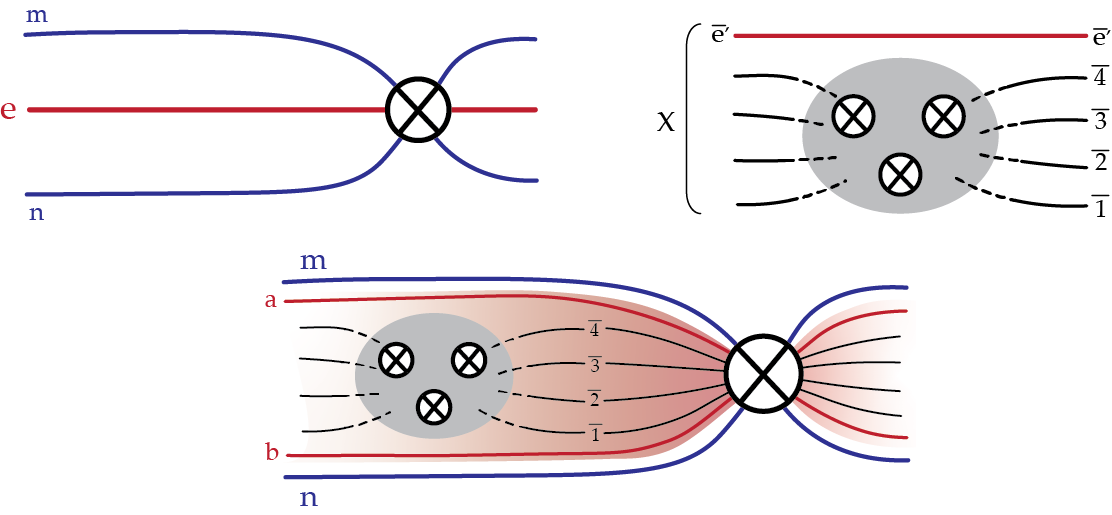}
    \caption{Blow up when $X$ is non-orientable.}
    \label{lastfigure}
\end{figure}

\section{From $2$-vertex graphs to bipartite graphs}\label{bipartite}

As we announced in the introduction, the case of bipartite embedding schemes reduces to the case of 2-vertex embedding schemes. This is summarized in the next two lemmas, and completes the proof of the general case of Theorem~\ref{main theorem}. 

Recall that an \emph{allowable contraction} of $G$ is the result of adding an edge in a face of $G$ between two non-adjacent vertices on the same side of the bipartition and contracting it. Note that the embedding scheme we obtain after an allowable contraction of $G$ has the same genus as $G$.

\begin{lemma}
    Let $G$ and $H$ be embedding schemes such that $H$ is obtained after a sequence of allowable contractions of $G$. If $H$ admits a perfect cross-cap drawing, then so does $G$. 
\end{lemma}
\begin{proof}
    Let $G=H_0,H_1,H_2,\dots,H_k=H$ be the sequence of embedding schemes that we obtain after contracting auxiliary edges $i_1,i_2,\dots,i_k$ one by one, i.e., we obtain $H_j$ by adding the edge $i_j$ to $H_{j-1}$ and contracting it. To prove the lemma, it is sufficient to show that if $H_1$ has a perfect cross-cap drawing then so does $G$. 

Let $f:(e_1,v_1,\dots,e_l,v_l)$ be the face of $G$ that is divided to two faces $(e_1,v_1,\dots,e_t,v_t,i_1,v_l)$ and $(i_1,v_t,e_{t+1},v_{t+1}\dots, e_l,v_l)$ by adding the auxiliary edge $i_1=v_tv_l$. After contracting $i_1$, we obtain $H_1$ with faces $f_1:( e_1,v_1,\dots,e_t,v*)$ and $f_2:(e_{t+1},v_{t+1},\dots,e_l,v*)$ where $v*$ is the vertex corresponding to the vertices $v_t$ and $v_l$ after the contraction. Let $\phi$ be a perfect cross-cap drawing for $H_1$. We uncontract the edge $i_1$ close to the vertex $v*$ and away from the cross-caps in $\phi$ and then we remove it. This turns $\phi$ into a perfect cross-cap drawing for $G$ and it finishes the proof. 
\end{proof}
\begin{lemma}
    For any connected bipartite embedding scheme $G$, there exists an embedding scheme $H$ that is obtained after a sequence of allowable contractions of $G$ such that $H$ has only two vertices and does not contain any loop.
\end{lemma}
\begin{proof}
    Let $V_A$ and $V_B$ be a partition of vertices of $G$ such that each edge of $G$ has one end in $V_A$ and the other in $V_B$. If each of these sets has only one element, then $G$ has only two vertices and we have nothing to prove. Let us assume that at least one of these sets has more than one element. Since $G$ is connected, there exists a face $f:(v_0,e_1,v_1,e_2,v_2,\dots)$ such that $v_0\neq v_2$ are two vertices in one of the sets, say $V_A$, and $v_1$ belongs to $V_B$. We add an auxiliary edge $i=v_0v_2$ in the face $f$ and contract it. This reduces the vertices in $V_A$ by one. The graph we obtain after this allowable contraction is still connected and bipartite. We continue this process until we obtain a 2-vertex embedding scheme. Since we only add auxiliary edges between the vertices in the same set $V_A$ or $V_B$, and all the edges of $G$ have one end in $V_A$ and the other in $V_B$, we do not create any loops during these allowable contractions. This finishes the proof.
\end{proof}

Therefore, starting from a bipartite embedding scheme $G$, if a sequence of allowable contractions yields a 2-vertex embedding scheme whose reduced graph is not one of the two graphs in Figure~\ref{reduced badly}, then by Theorem~\ref{main theorem} for $2$-vertex graphs, proved in the previous section, $G$ admits a perfect cross-cap drawing. This concludes the proof of Theorem~\ref{main theorem}

\subsubsection*{Acknowledgements} We are very grateful to the anonymous reviewers for helpful comments. 

This work is part of the first named author’s PhD thesis at Université Gustave Eiffel, under the joint supervision of the second and third named authors.

\bibliographystyle{elsarticle-num}
\bibliography{bib.bib}

\end{document}